\documentclass[11pt,oneside,reqno]{amsart}
\usepackage{amstext}
\usepackage{amsthm}
\usepackage[latin9]{inputenc}
\usepackage{color}
\definecolor{note_fontcolor}{rgb}{0.800781, 0.800781, 0.800781}
\usepackage{verbatim}
\usepackage{longtable}

\makeatletter

\providecommand{\tabularnewline}{\\}

\numberwithin{equation}{section}
\numberwithin{figure}{section}
\theoremstyle{plain}
\newtheorem{thm}{\protect\theoremname}
 \theoremstyle{definition}
\newtheorem{defn}[thm]{\protect\definitionname}
\theoremstyle{definition}

\newtheorem{cmnt}[thm]{Comment}
\newtheorem{remark}[thm]{Remark}
\theoremstyle{plain}
\newtheorem{prop}[thm]{\protect\propositionname}
\theoremstyle{plain}
\newtheorem{lem}[thm]{\protect\lemmaname}

\@ifundefined{date}{}{\date{}}

\usepackage{pdfsync}
\def\QED{\hskip0.1em\hfill\null\ \null\nobreak\hfill
\kern3pt\lower1.8pt\vbox{\hrule\hbox
{\vrule\kern1pt\vbox{\kern1.7pt \hbox{$\scriptstyle
QED$}\kern0.2pt}\kern1pt\vrule}\hrule}}

\numberwithin{equation}{section}

\renewcommand{\mathcal}[1]{\mathscr{#1}}

\newcommand{\ie}{\textit{i.e.}}
\newcommand{\eg}{\textit{e.g.}}

\usepackage{mathrsfs}
\usepackage[lining,tabular]{fbb}
\usepackage[libertine,bigdelims]{newtxmath}
\usepackage[cal=boondoxo,bb=boondox,frak=boondox]{mathalfa}

\newcommand{\shortrule}{\rule[0.5ex]{1cm}{0.5pt}}

\newcommand{\InCoord}[1]{
\medskip

\centerline{\shortrule$\blacktriangledown\blacktriangledown\blacktriangledown$\shortrule}

{\footnotesize #1}

\centerline{\shortrule$\blacktriangle\blacktriangle\blacktriangle$\shortrule}
\medskip
}

\usepackage[all]{xy}
\newcommand{\heb}[1]{}

\definecolor{darkocre}{RGB}{121,51,12} 
\definecolor{darkgreen}{RGB}{0,100,20}

\usepackage[normalem]{ulem}

\newcommand{\strn}{deformation jet}
\newcommand{\vjet}{velocity jet}

\makeatother

  \providecommand{\definitionname}{Definition}
  \providecommand{\examplename}{Example}
  \providecommand{\lemmaname}{Lemma}
  \providecommand{\propositionname}{Proposition}
\providecommand{\theoremname}{Theorem}

\begin{document}

\global\long\def\st{X}
\global\long\def\bdl{Y}
\global\long\def\pr{\pi}
\global\long\def\lgr{\mathcal{L}}
\global\long\def\base{\mathcal{B}}
\global\long\def\Hom{\mathrm{Hom}}


\global\long\def\bolda{\boldsymbol{a}}
 \global\long\def\boldb{\boldsymbol{b}}
 \global\long\def\boldc{\boldsymbol{c}}
 \global\long\def\boldd{\boldsymbol{d}}
 \global\long\def\bolde{\boldsymbol{e}}
 \global\long\def\boldf{\boldsymbol{f}}
 \global\long\def\boldg{\boldsymbol{g}}
 \global\long\def\boldh{\boldsymbol{h}}
 \global\long\def\boldi{\boldsymbol{i}}
 \global\long\def\boldj{\boldsymbol{j}}
 \global\long\def\boldk{\boldsymbol{k}}
 \global\long\def\boldl{\boldsymbol{l}}
 \global\long\def\boldm{\boldsymbol{m}}
 \global\long\def\boldn{\boldsymbol{n}}
 \global\long\def\boldo{\boldsymbol{o}}
 \global\long\def\boldp{\boldsymbol{p}}
 \global\long\def\boldq{\boldsymbol{q}}
 \global\long\def\boldr{\boldsymbol{r}}
 \global\long\def\bolds{\boldsymbol{s}}
 \global\long\def\boldt{\boldsymbol{t}}
 \global\long\def\boldu{\boldsymbol{u}}
 \global\long\def\boldv{\boldsymbol{v}}
 \global\long\def\boldw{\boldsymbol{w}}
 \global\long\def\boldx{\boldsymbol{x}}
 \global\long\def\boldy{\boldsymbol{y}}
 \global\long\def\boldz{\boldsymbol{z}}

\global\long\def\boldA{\boldsymbol{A}}
 \global\long\def\boldB{\boldsymbol{B}}
 \global\long\def\boldC{\boldsymbol{C}}
 \global\long\def\boldD{\boldsymbol{D}}
 \global\long\def\boldE{\boldsymbol{E}}
 \global\long\def\boldF{\boldsymbol{F}}
 \global\long\def\boldG{\boldsymbol{G}}
 \global\long\def\boldH{\boldsymbol{H}}
 \global\long\def\boldI{\boldsymbol{I}}
 \global\long\def\boldJ{\boldsymbol{J}}
 \global\long\def\boldK{\boldsymbol{K}}
 \global\long\def\boldL{\boldsymbol{L}}
 \global\long\def\boldM{\boldsymbol{M}}
 \global\long\def\boldN{\boldsymbol{N}}
 \global\long\def\boldO{\boldsymbol{O}}
 \global\long\def\boldP{\boldsymbol{P}}
 \global\long\def\boldQ{\boldsymbol{Q}}
 \global\long\def\boldR{\boldsymbol{R}}
 \global\long\def\boldS{\boldsymbol{S}}
 \global\long\def\boldT{\boldsymbol{T}}
 \global\long\def\boldU{\boldsymbol{U}}
 \global\long\def\boldV{\boldsymbol{V}}
 \global\long\def\boldW{\boldsymbol{W}}
 \global\long\def\boldX{\boldsymbol{X}}
 \global\long\def\boldY{\boldsymbol{Y}}
 \global\long\def\boldZ{\boldsymbol{Z}}

\global\long\def\calA{{\mathcal{A}}}
 \global\long\def\calB{{\mathcal{B}}}
 \global\long\def\calC{{\mathcal{C}}}
 \global\long\def\calD{{\mathcal{D}}}
 \global\long\def\calE{{\mathcal{E}}}
 \global\long\def\calF{{\mathcal{F}}}
 \global\long\def\calG{{\mathcal{G}}}
 \global\long\def\calH{{\mathcal{H}}}
 \global\long\def\calI{{\mathcal{I}}}
 \global\long\def\calJ{{\mathcal{J}}}
 \global\long\def\calK{{\mathcal{K}}}
 \global\long\def\calL{{\mathcal{L}}}
 \global\long\def\calM{{\mathcal{M}}}
 \global\long\def\calN{{\mathcal{N}}}
 \global\long\def\calO{{\mathcal{O}}}
 \global\long\def\calP{{\mathcal{P}}}
 \global\long\def\calQ{{\mathcal{Q}}}
 \global\long\def\calR{{\mathcal{R}}}
 \global\long\def\calS{{\mathcal{S}}}
 \global\long\def\calT{{\mathcal{T}}}
 \global\long\def\calU{{\mathcal{U}}}
 \global\long\def\calV{{\mathcal{V}}}
 \global\long\def\calW{{\mathcal{W}}}
 \global\long\def\calX{{\mathcal{X}}}
 \global\long\def\calY{{\mathcal{Y}}}
 \global\long\def\calZ{{\mathcal{Z}}}

\global\long\def\frakA{\mathfrak{A}}
 \global\long\def\frakB{\mathfrak{B}}
 \global\long\def\frakC{\mathfrak{C}}
 \global\long\def\frakD{\mathfrak{D}}
 \global\long\def\frakE{\mathfrak{E}}
 \global\long\def\frakF{\mathfrak{F}}
 \global\long\def\frakG{\mathfrak{G}}
 \global\long\def\frakH{\mathfrak{H}}
 \global\long\def\frakI{\mathfrak{I}}
 \global\long\def\frakJ{\mathfrak{J}}
 \global\long\def\frakK{\mathfrak{K}}
 \global\long\def\frakL{\mathfrak{L}}
 \global\long\def\frakM{\mathfrak{M}}
 \global\long\def\frakN{\mathfrak{N}}
 \global\long\def\frakO{\mathfrak{O}}
 \global\long\def\frakP{\mathfrak{P}}
 \global\long\def\frakQ{\mathfrak{Q}}
 \global\long\def\frakR{\mathfrak{R}}
 \global\long\def\frakS{\mathfrak{S}}
 \global\long\def\frakT{\mathfrak{T}}
 \global\long\def\frakU{\mathfrak{U}}
 \global\long\def\frakV{\mathfrak{V}}
 \global\long\def\frakW{\mathfrak{W}}
 \global\long\def\frakX{\mathfrak{X}}
 \global\long\def\frakY{\mathfrak{Y}}
 \global\long\def\frakZ{\mathfrak{Z}}

\global\long\def\fraka{\mathfrak{a}}
 \global\long\def\frakb{\mathfrak{b}}
 \global\long\def\frakc{\mathfrak{c}}
 \global\long\def\frakd{\mathfrak{d}}
 \global\long\def\frake{\mathfrak{e}}
 \global\long\def\frakf{\mathfrak{f}}
 \global\long\def\frakg{\mathfrak{g}}
 \global\long\def\frakh{\mathfrak{h}}
 \global\long\def\fraki{\mathfrak{i}}
 \global\long\def\frakj{\mathfrak{j}}
 \global\long\def\frakk{\mathfrak{k}}
 \global\long\def\frakl{\mathfrak{l}}
 \global\long\def\frakm{\mathfrak{m}}
 \global\long\def\frakn{\mathfrak{n}}
 \global\long\def\frako{\mathfrak{o}}
 \global\long\def\frakp{\mathfrak{p}}
 \global\long\def\frakq{\mathfrak{q}}
 \global\long\def\frakr{\mathfrak{r}}
 \global\long\def\fraks{\mathfrak{s}}
 \global\long\def\frakt{\mathfrak{t}}
 \global\long\def\fraku{\mathfrak{u}}
 \global\long\def\frakv{\mathfrak{v}}
 \global\long\def\frakw{\mathfrak{w}}
 \global\long\def\frakx{\mathfrak{x}}
 \global\long\def\fraky{\mathfrak{y}}
 \global\long\def\frakz{\mathfrak{z}}

\global\long\def\bbA{{\mathbb{A}}}
 \global\long\def\bbB{{\mathbb{B}}}
 \global\long\def\bbC{{\mathbb{C}}}
 \global\long\def\bbD{{\mathbb{D}}}
 \global\long\def\bbE{{\mathbb{E}}}
 \global\long\def\bbF{{\mathbb{F}}}
 \global\long\def\bbG{{\mathbb{G}}}
 \global\long\def\bbH{{\mathbb{H}}}
 \global\long\def\bbI{{\mathbb{I}}}
 \global\long\def\bbJ{{\mathbb{J}}}
 \global\long\def\bbK{{\mathbb{K}}}
 \global\long\def\bbL{{\mathbb{L}}}
 \global\long\def\bbM{{\mathbb{M}}}
 \global\long\def\bbN{{\mathbb{N}}}
 \global\long\def\bbO{{\mathbb{O}}}
 \global\long\def\bbP{{\mathbb{P}}}
 \global\long\def\bbQ{{\mathbb{Q}}}
 \global\long\def\bbR{{\mathbb{R}}}
 \global\long\def\bbS{{\mathbb{S}}}
 \global\long\def\bbT{{\mathbb{T}}}
 \global\long\def\bbU{{\mathbb{U}}}
 \global\long\def\bbV{{\mathbb{V}}}
 \global\long\def\bbW{{\mathbb{W}}}
 \global\long\def\bbX{{\mathbb{X}}}
 \global\long\def\bbY{{\mathbb{Y}}}
 \global\long\def\bbZ{{\mathbb{Z}}}
\global\long\def\scrA{\mathscr{A}}
 \global\long\def\scrB{\mathscr{B}}
 \global\long\def\scrC{\mathscr{C}}
 \global\long\def\scrD{\mathscr{D}}
 \global\long\def\scrE{\mathscr{E}}
 \global\long\def\scrF{\mathscr{F}}
 \global\long\def\scrG{\mathscr{G}}
 \global\long\def\scrH{\mathscr{H}}
 \global\long\def\scrI{\mathscr{I}}
 \global\long\def\scrJ{\mathscr{J}}
 \global\long\def\scrK{\mathscr{K}}
 \global\long\def\scrL{\mathscr{L}}
 \global\long\def\scrM{\mathscr{M}}
 \global\long\def\scrN{\mathscr{N}}
 \global\long\def\scrO{\mathscr{O}}
 \global\long\def\scrP{\mathscr{P}}
 \global\long\def\scrQ{\mathscr{Q}}
 \global\long\def\scrR{\mathscr{R}}
 \global\long\def\scrS{\mathscr{S}}
 \global\long\def\scrT{\mathscr{T}}
 \global\long\def\scrU{\mathscr{U}}
 \global\long\def\scrV{\mathscr{V}}
 \global\long\def\scrW{\mathscr{W}}
 \global\long\def\scrX{\mathscr{X}}
 \global\long\def\scrY{\mathscr{Y}}
 \global\long\def\scrZ{\mathscr{Z}}

\global\long\def\brk#1{\left(#1\right)}
\global\long\def\Brk#1{\left[#1\right]}
\global\long\def\BRK#1{\left\{  #1\right\}  }
\global\long\def\Average#1{\left\langle #1\right\rangle }
\global\long\def\sAverage#1{\langle#1\rangle}
\global\long\def\mean#1{\overline{#1}}
\global\long\def\Floor#1{\lfloor#1 \rfloor}
 \global\long\def\Ceil#1{\lceil#1 \rceil}
 \global\long\def\Abs#1{\left| #1 \right|}
\global\long\def\Norm#1{\left\Vert #1 \right\Vert }
\global\long\def\mymat#1{\begin{pmatrix} #1 \end{pmatrix}}
 \global\long\def\Cases#1{\begin{cases}
 #1 \end{cases}}

\global\long\def\deriv#1#2{\frac{d#1}{d#2}}
 \global\long\def\dderiv#1#2{\frac{d^{2}#1}{d#2^{2}}}

\global\long\def\pd#1#2{\frac{\partial#1}{\partial#2}}
 \global\long\def\pdd#1#2{\frac{\partial^{2}#1}{\partial#2^{2}}}
 \global\long\def\pddm#1#2#3{\frac{\partial^{2}#1}{\partial#2\,\partial#3}}
 \global\long\def\pddd#1#2{\frac{\partial^{3}#1}{\partial#2^{3}}}
 \global\long\def\pdddd#1#2{\frac{\partial^{4}#1}{\partial#2^{4}}}

\global\long\def\bsplit{\begin{split}\end{split}
 }
 \global\long\def\esplit{{split}}
 \global\long\def\baligned{\begin{aligned}\end{aligned}
 }
 \global\long\def\ealigned{{aligned}}

\global\long\def\Emph#1{{\slshape\bfseries#1}}
\global\long\def\EMPH#1{{\bfseries\fontfamily{pzc}\large#1}}

\global\long\def\vecsim#1{\underset{\sim}{\boldsymbol{#1}}}
 \global\long\def\matapprox#1{\underset{\approx}{\boldsymbol{#1}}}
 \global\long\def\Vector#1{\smash{\underset{\sim}{\boldsymbol{#1}}}}
 \global\long\def\Matrix#1{\smash{\underset{\approx}{\boldsymbol{#1}}}}
 \global\long\def\Bold#1{\boldsymbol{#1}}
 \providecommand{\bs}[1]{\boldsymbol{#1}}

\providecommand{\grad}{\boldsymbol{\nabla}} \global\long\def\define{\stackrel{\text{def}}{=}}
 \global\long\def\defi{\stackrel{\text{def}}{=}}
 \providecommand{\half}{\frac{1}{2}} \global\long\def\smallhalf{\tfrac{1}{2}}
 \global\long\def\Implies{\Longrightarrow}
 \global\long\def\Iff{\Longleftrightarrow}
 \global\long\def\R{\mathbb{R}}
 \global\long\def\sequence#1#2#3{{#1}_{#2},\ldots,{#1}_{#3}}
 \global\long\def\textand{\quad\text{ and }\quad}
 \global\long\def\Textand{\qquad\text{ and }\qquad}

\global\long\def\var{{\mathop{}\!\mathrm{Var}}}
 \global\long\def\cov{{\mathop{}\!\mathrm{Cov}}}
 \global\long\def\Bas{{\mathop{}\!\mathrm{Bas}}}
 \global\long\def\argmax{\operatornamewithlimits{arg\ max}}
 \global\long\def\argmin{\operatornamewithlimits{arg\ min}}
 \global\long\def\sgn{{\mathop{}\!\mathrm{sgn}}}
 \providecommand{\rank}{{\mathop{}\!\mathrm{rank}}}

\newcommand{\divergence}{\operatorname{div}}
 \global\long\def\Vol{dX}

\global\long\def\S{\mathcal{S}}
 \global\long\def\vp{\varphi}
 \global\long\def\e{\varepsilon}
 \global\long\def\id{\mathop{}\!\mathrm{id}}
 \global\long\def\Q{\mathcal{Q}}

\global\long\def\Image{\mathop{}\!\mathrm{Image}}
 \global\long\def\Hom{\mathop{}\!\mathrm{Hom}}
 \global\long\def\Supp{\mathop{}\mathrm{Support}}

\global\long\def\vals#1{\mathbb{#1}}
\global\long\def\vbm#1{\overrightarrow{#1}}
{} \global\long\def\strain{\mathcal{E}}
\global\long\def\iotaF{P}
 \global\long\def\ForceDensity{\mathbf{b}}
 \global\long\def\LoadingDensity{\mathbf{B}}
 \global\long\def\SurfaceForceDensity{\mathbf{t}}
 \global\long\def\SurfaceLoadingDensity{\mathbf{T}}
 \global\long\def\StressDensity{\mathbf{s}}
 \global\long\def\crel{\Psi}
\global\long\def\ConstitutiveDensity{\mathrm{\psi}}
\global\long\def\sten{U}

\global\long\def\inclusion{\mathcal{I}}
\global\long\def\trst{\tau}
\global\long\def\trstf{\tau}
\global\long\def\subb{\mathcal{P}}
\global\long\def\subm{\mathcal{V}}

\global\long\def\bk{\kappa}
 \global\long\def\jk{j^{1}\bk}
 \global\long\def\dxB{dx_{\B}}
 \global\long\def\dxY{dx_{Y}}
 \global\long\def\dxJY{dx_{J^{1}Y}}
 \global\long\def\parB{\partial^{\B}}
 \global\long\def\parY{\partial^{Y}}
 \global\long\def\parJY{\partial^{J^{1}Y}}

\global\long\def\Cone{C^{1}}
 \global\long\def\Czero{C^{0}}
 \global\long\def\Cinfty{C^{\infty}}
\global\long\def\B{\mathcal{B}}
\global\long\def\LamD{\Lambda^{d}T^{*}\B}
\global\long\def\man{M}
\global\long\def\pnt{m}

\global\long\def\push{\sharp}
\global\long\def\pushf{\circ}

\global\long\def\sec#1{\wh#1}
\global\long\def\difo#1{\breve{#1}}
\global\long\def\sdv{\StressDensity}
\global\long\def\s{\zeta}

\global\long\def\inc{\,\lrcorner\,}
\global\long\def\rep#1{\overline{#1}}
\global\long\def\avf{u}
 \global\long\def\inj{\hookrightarrow}
\global\long\def\ext{{\textstyle \Lambda}}
\global\long\def\isom{\cong}
\global\long\def\vs{\mathbf{W}}
\global\long\def\fall{,\quad\text{for all}\quad}
\global\long\def\reals{\mathbb{R}}
\global\long\def\resto#1{|_{#1}}
\global\long\def\tto{\longrightarrow}
\global\long\def\wh#1{\widehat{#1}}
\global\long\def\lmt{\longmapsto}

\global\long\def\diver{\divergence}
\global\long\def\divr{\mathop{}\mathrm{div}}
\global\long\def\stds{\mathfrak{S}}
\global\long\def\trstds{\mathfrak{T}}{} 
\global\long\def\bfcs{\mathbb{G}}
\global\long\def\contr{\mathbb{C}}
\global\long\def\cres{\mathcal{C}}

\title{Stress Theory for Classical Fields}

\author{Raz Kupferman$ {}^* $, Elihu Olami$ {}^* $ }
\address{$ {}^* $Einstein Institute of Mathematics,
The Hebrew University of Jerusalem
Jerusalem, 9190401, Israel}

\author{Reuven Segev$ {}^{**} $}
\address{$ {}^{**} $Department of Mechanical Engineering, Ben-Gurion
University of the Negev, Beer-Sheva 84105, Israel}

\date{\today}
\begin{abstract}

Classical field theories together with the Lagrangian and Eulerian approaches to continuum mechanics are embraced under a geometric setting of a fiber bundle. The base manifold can be either the body manifold of continuum mechanics, space manifold, or space-time. Differentiable sections of the fiber bundle represent configurations of the system and the configuration space containing them is given the structure of an infinite dimensional manifold. Elements of the cotangent bundle of the configuration space are interpreted as generalized forces and a representation theorem implies that there exist a stress object representing forces, non-uniquely. The properties of stresses are studies as well as the role of constitutive relations in the present general setting.

\smallskip
MSC codes: 74A10, 53Z05, 74A60
\bigskip{}
\end{abstract}

\maketitle

\tableofcontents


\section{Introduction}

Physical theories for which the states are represented {by  sections
of a fiber bundle} are predominant in both classical field theories
of theoretical physics and studies of the material structure of bodies
in continuum mechanics. This paper is concerned with the corresponding
stress theory.

For about half a century now, classical fields are modeled mathematically
in the theoretical physics literature as sections of fiber bundles
over space-time. Since the pioneering works on modern formulations
of classical field theories, for example, \cite{Trautman1967,Komorowski1968,Komorowski1969,Sniatycki70,Hermann1970,Krupka1971},
a generic field is viewed as a section $\bk:\base\to\bdl$ of
a fiber bundle $\pr:\bdl\to\base$ for a $d$-dimensional space-time
$\base$. The field equations are obtained by considering stationary
values of an action integral 
\begin{equation}
\int_{\base}\lgr(j^{r}\bk)
\end{equation}
where the Lagrangian function $\lgr:J^{r}Y\to\ext^{d}T^{*}\base$
is a fiber preserving mapping of $r$-jets into the bundle of $d$-alternating
multilinear forms over $\base$. (See, for example, \cite{kijowski1979,deLeon1985,Binz1988,Echeverria1996,Ramond2001,Gotay2003,Giachetta2009,Frankel2012}.)
The variational analysis of the action integral yields terms that
may be interpreted as the components of the stress tensor.

Nontrivial fiber bundles appeared in continuum mechanics originally
in works considering dislocations, and material uniformity and homogeneity.
(See, \cite{Bilby55,Kondo1955,Nol67,Wan67,Bloom1979,epselz}.) The
modern formulations of these theories usually consider sections of
the principal bundle of frames, or moving frames, over the body manifold
as a mathematical model for the distributions of material directions.
Since the sections are defined over the body manifold, no reference
should be made to the physical space and its conceivable Euclidean
structure.

Formulation of continuum mechanics using sections of a general fiber
bundle has an additional advantage. While a major portion of studies
in continuum mechanics use the Lagrangian approach in which material
points serve as fundamental objects, the Eulerian viewpoint may be
advantageous for studies of chemically reacting matter and growing
bodies. The Lagrangian and Eulerian viewpoints are unified when continuum
theories are modeled on fiber bundles. For the Lagrangian formulation,
one simply considers the trivial bundle $\base\times{S}\to\base$
in which ${S}$ is the manifold representing the ambient physical
space. 
In the Eulerian picture, $ \base $ is interpreted as the space manifold or a region therein.

Modern studies of the mathematical structure of stress theory in continuum
mechanics may be traced back to \cite{Noll59} and subsequently \cite{Gurtin1967,Gurtin1975,Silhavy1985,Silhavy1991,Silhavy2008},
for example. The relevance of the notion of stress to field theories
led to contributions from the physics community, for example, \cite[p.~168]{kijowski1979},
which, in some cases, applied ideas originating in the continuum mechanics
research (see \cite[p.~168]{Hehl1976,kijowski1979,Hehl1986,Hehl1991,Hehl1995,Hehl1997}).

The stress object emerges in field theories as {the vertical derivative of the Lagrangian function}. Yet, the studies of the stress object and the field equations it should satisfy are
relevant in the more general situation where a Lagrangian mapping
is not readily available. In the continuum theory of dislocations,
for example, it is hard to expect that the motion of dislocations
will be governed by a potential. 

Thus, this paper considers the stress object and the equations governing
it for fields represented as sections of a general fiber bundle.
Extending the terminology in \cite{TruesdellToupin60}, and in view
of the applications described above, we will use the terminology a
\Emph{classical field theory} to refer to any such  mechanical,
or other physical, theory.

In our approach, the analysis of the stress field is put in a broader (global) context. Extending \cite{Segev1986}, we consider a configuration space of sections which is an infinite dimensional manifold. Generalized forces are viewed as elements of the cotangent bundle of the configuration space and stresses emerge from a representation theorem for the force linear functionals.  This ``weak'' approach allows for generalized, singular, stress fields, with corresponding distributional field equations. We give special attention to the case of smooth stress fields for which  we write down the field equation in an explicit differential form.
For example, a
weak formulation of $p$-form, \cite{Henneaux1986,Henneaux1988},
premetric, \cite{Hehl-Itin06}, electrodynamics was shown in \cite{Segev_e_d_2016}
to follow from stress theory for fields represented by $p$-forms
in the case where the stress object has a particularly simple form. It should be mentioned that we study here the theory concerning the existence of stresses and the equations it satisfies; we do not study the analytic aspects of the field equations obtained after the constitutive relations are used, in tems of existence and uniqueness of solutions, appropriate function spaces, etc.


Following the introduction in Section 2 of the notation and terminology
used pertaining to fiber bundles and their jet bundles, Section 3
is concerned with the infinite dimensional configuration space of
sections. Generalized velocities and generalized forces are modeled
as elements of the tangent and cotangent bundles of the configuration
space, respectively. Particular attention is given to smooth forces,
those given in terms of body forces and surface forces. Section 4
considers the analog of ``local configurations'' as in \cite[pp. 51--52]{Truesdell1965}.
In the present general geometric setting, these are described by sections
of the jet bundle. Next, local velocities and their relation to the
jets of generalized velocity fields are discussed. Stresses are considered
in Section 5. Variational stresses are defined as functionals conjugate
to velocity jets. A representation theorem for generalized forces
in terms of variational stresses relate the two type of objects through
a general version of the principle of virtual power. While, in general,
variational stresses are tensor-valued measures, smooth variational
stresses, those represented by smooth tensor valued densities induce
traction stresses. The traction stress object determines surface forces
on oriented $(d-1)$-submanifolds via a generalization of Cauchy's
formula. A differential operator generalizing the traditional divergence
of the stress tensor is defined next, enabling one to write a generalized
version of differential equations of equilibrium and boundary conditions.
It is observed, that while in the classical formulation of stress
theory, the stress tensor both acts on the rate of change of the deformation
gradient to produce power and determines the traction on hypersurfaces,
in the geometry of general manifolds, two objects, the variational
stress and the traction stresses, are needed for these two roles.
While the variational stress determines a unique traction stress,
a traction stress field do not determine a unique variational stress.
Next, in Section \ref{sec:Constitutive}, loadings and constitutive relations
are introduced leading to a formulation of the problem of stress analysis.
Finally, a number of particular cases are presented in Section \ref{sec:Examples}.
In particular, the relation between stress theory and premetric $p$-form
electrodynamics is summarized.


\section{Preliminaries and Notation}

The fundamental geometric object considered in this work is that of
a fiber bundle, the sections of which are identified with the configurations of a mechanical
system or with classical fields. In this section we introduce the notation and
terminology adopted throughout this paper. 

A fiber bundle \cite{Saunders} will be denoted by a triple $(Y,\pi,\man)$, where $Y$ is the total space, $\man$ is the base manifold and $\pi:Y\to\man$ is the projection.
Let $\man$ be a manifold. We denote by 
\[
(T\man,\tau_{\man},\man)\Textand(T^*\man,\tau_{\man}^*,\man)
\]
its tangent and cotangent bundles. The bundle of $k$-alternating
multilinear forms will be denoted by $(\ext^{k}T^*\man,\tau_{\man}^*,\man)$.

Let $(Y,\pi,\man)$ be a fiber bundle. For a section
$s:\man\to Y$ and a point $m\in \man$, we denote by $s_{m}$,  rather than $s(m)$, the value
of $s$ at $m$. We also denote the fiber of $Y$ at $m$ by $Y_m:=\pi^{-1}(m)$.

Consider $T\pi:TY\to T\man$. The kernel of $T\pi$ in $TY$ is commonly denoted
by $V\pi\subset TY$, and is referred to as the \Emph{vertical sub-bundle} of $TY$. The
set $V\pi$ is the total space of two bundles: the vector bundle 
\[
(V\pi,\tau_{Y}|_{V\pi},Y),
\]
and the fiber bundle 
\[
(V\pi,\pi\circ\tau_{Y}|_{V\pi},\man).
\]

\[
\begin{xy}(20,0)*+{\man}="X";
(0,0)*+{Y}="Y";(20,20)*+{T\man}="TX";
(0,20)*+{TY}="TY";
(-20,20)*+{V\pi}="VY";
{\ar@{->}_{\pi}"Y";"X"};
{\ar@{->}^{T\pi}"TY";"TX"};
{\ar@{->}^{\tau_{\man}}"TX";"X"};{\ar@{->}_{\tau_{Y}}"TY";"Y"};
{\ar@{->}_{\tau_{Y}|_{V\pi}}"VY";"Y"};
{\ar@{^{(}->}^{\text{incl.}}"VY";"TY"};
\end{xy}
\]

The vertical bundle of $Y$ is often denoted in the literature by $VY$ rather than by $V\pi$. The notation $VY$ may be ambiguous in instances where $Y$ is the total space of multiple bundles;  the notation $V\pi$ makes explicit the projection with respect to which verticality is defined. On the other hand, the latter notation is often cumbersome, for example, when the projection is a composition of several projections, some of which restricted to sub-bundles. For improved readability we adopt the following notation scheme: for a fiber bundle $(Y,\pi_Y,\man)$, we denote its vertical bundle by $VY$ in cases where $Y$ is the total space of a single bundle, or, in the case of repeated projections, $Y\to Z\to\cdots \to \man$, in which case verticality is implied relative to the projection onto the manifold $\man$. 

Consider two fiber bundles $(Y,\pi_Y,\man)$ and $(Z,\pi_Z,\man)$ over
the same base manifold $\man$. Let $\vp:Y\to Z$ be a fiber bundle morphism, i.e., $\pi_Z\circ\vp = \pi_Y$.
The restriction of the tangent map $T\vp:TY\to TZ$ to $VY$ defines
a vertical bundle morphism, 
\[
T\vp|_{VY}:VY\tto VZ.
\]

\[
\begin{xy}
(0,0)*+{\man} = "X";%
(-15,15)*+{Y} = "Y";
(15,15)*+{Z} = "VY";%
(-15,30)*+{VY} = "J1Y";%
(15,30)*+{VZ} = "VJ1Y";%
{\ar@{->}_{\pi_Y} "Y"; "X"};%
{\ar@{->}^{\pi_Z} "VY"; "X"};%
{\ar@{->}^{\vp} "Y"; "VY"};%
{\ar@{->}_{\tau_Y|_{VY}} "J1Y"; "Y"};%
{\ar@{->}^{\tau_Z|_{VZ}} "VJ1Y"; "VY"};%
{\ar@{->}^{T\vp|_{VY}} "J1Y"; "VJ1Y"};%
\end{xy} 
\]

Let $(Y,\pi,N)$ be a fiber bundle and let $f:\man\to N$ be a differentiable
mapping. One has the natural pullback bundle, $(f^*Y,f^*\pi,\man)$
with $(f^*Y)_{m}$ canonically identified with $Y_{f(m)}$ via the bundle morphism $\pi^*f:f^*Y\to Y$ over $f$,
as in the  diagram below. A section $s:N\to Y$ induces the section
$f^*s:\man\to f^*Y$ satisfying $\pi^*f\circ f^*s=s\circ f$;  in other words,  $(f^*s)_m$ is identified with $s_{f(m)}$.

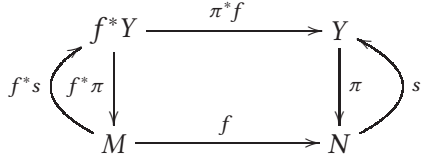
\[
\begin{xy}
(-15,15)*+{\man} = "Y";
(15,15)*+{N} = "VY";%
(-15,30)*+{f^*Y} = "J1Y";%
(15,30)*+{Y} = "VJ1Y";%
{\ar@{->}^{f} "Y"; "VY"};%
{\ar@{->}_{f^*\pi} "J1Y"; "Y"};%
{\ar@{->}^{\pi} "VJ1Y"; "VY"};%
{\ar@{->}^{\pi^*f} "J1Y"; "VJ1Y"};%
{\ar@{->}@/^{2pc}/^{{f^*s}} "Y"; "J1Y"};
{\ar@{->}@/_{2pc}/_{{s}} "VY"; "VJ1Y"};
\end{xy} 
\]

The tangent map of $f$, $Tf:T\man\to TN$ induces the differential of $f$, $df: T\man\to f^*TN$ satisfying $\tau_N^*f\circ df=Tf$.

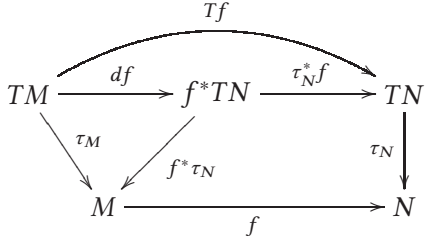
\[
\begin{xy}
(0,0)*+{\man} = "M";
(40,0)*+{N} = "N";
(40,15)*+{TN} ="TN";
(-10,15)*+{T\man} = "TM";
(15,15)*+{f^*TN} ="fTN";
{\ar@{->}_{f} "M"; "N"};
{\ar@{->}^{\tau_\man} "TM"; "M"};
{\ar@{->}^{f^*\tau_N} "fTN"; "M"};
{\ar@{->}_{\tau_N} "TN"; "N"};
{\ar@{->}@/^{2pc}/^{Tf} "TM"; "TN"};
{\ar@{->}^{df} "TM"; "fTN"};
{\ar@{->}^{\tau_N^*f} "fTN"; "TN"};

\end{xy}
\]

In particular, for two fiber bundles $(Y,\pi,\man)$ and $(Z,\rho,\man)$ over
the same base manifold $\man$, and a fiber bundle morphism $\vp:Y\to Z$,
 we use 
\begin{equation}
d\vp|_{VY}:VY\tto \vp^*VZ
\label{eq:vert_d}
\end{equation}
to denote the vertical derivative of $\vp$, i.e., the restriction of $d\vp$ to $VY$; this vertical derivative is sometimes denoted in the literature by $\delta\vp$.

Let $(Y,\pi_Y,N)$ and $(Z,\pi_Z,N)$ be fiber bundles over $N$. Let $f:\man\to N$ and let $\vp:Y\to Z$ be a fiber bundle morphism. Then, $f$ induces a fiber bundle morphism
$f^*\vp:f^*Y\to f^*Z$, defined by the equality $\pi_Z^*f\circ f^*\vp=\vp\circ \pi_Y^*f$.

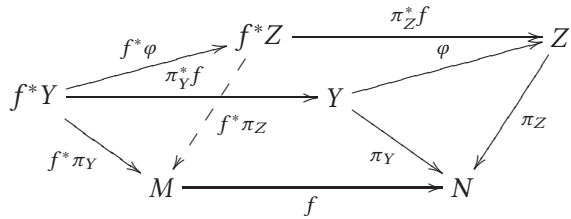
 \[
\begin{xy}
(-28,20)*+{\man} = "1x1";
(12,20)*+{N} = "2x1";%
(-45,32)*+{f^*Y} = "1x2";%
(-5,32)*+{Y} = "2x2";%
(-15,40)*+{f^*Z} = "1x3";
(25,40)*+{Z} = "2x3";%
{\ar@{<-}^{f} "2x1"; "1x1"};%
{\ar@{->}_{\pi_Y} "2x2"; "2x1"};%
{\ar@{->}^{\pi_Z} "2x3"; "2x1"};%
{\ar@{->}_{f^*\pi_Y} "1x2"; "1x1"};%
{\ar@{-->}^{f^*\pi_Z} "1x3"; "1x1"};%
{\ar@{<-}_{\vp} "2x3"; "2x2"};%
{\ar@{->}^{\pi_Y^*f} "1x2"; "2x2"};%
{\ar@{<-}_{f^*\vp} "1x3"; "1x2"};%
{\ar@{->}^{\pi_Z^*f} "1x3"; "2x3"};%
\end{xy} 
\]

{For a manifold $\man$, $\scrD(\man)$ denotes the space of  smooth real-valued functions on $\man$.} For a fiber bundle $(Y,\pi,\man)$,
it is common to denote by $C^{k}(\pi)$ the set of $C^{k}$-sections $\man\to Y$.
As in the case of the vertical bundle, we note that the space of $C^{k}$-sections $\man\to Y$ 
 is often denoted in the literature by $C^{k}(Y)$. Here too, we adopt the following notation scheme: for a fiber bundle $(Y,\pi_Y,\man)$, we denote its $C^{k}$-sections by $C^{k}(Y)$ in cases where $Y$ is the total space of a single bundle, or, in the case of repeated projections, $Y\to Z \to\cdots \to \man$, when the section is with respect to projection onto the manifold $\man$. 

Consider two fiber bundles $(Y,\pi_Y,\man)$ and $(Z,\pi_Z,\man)$ and let $\vp:Y\to Z$ be a fiber bundle morphism. Then $\vp$ induces a map between sections, 
\[
C^{k}(Y)\ni s\lmt \vp\circ s\in C^{k}(Z).
\]
This mapping is often denoted by $\vp_{*}:C^{k}(Y)\to C^{k}(Z)$, however, 
we will  write explicitly either $\vp\pushf s$ or just $\vp(s)$. 

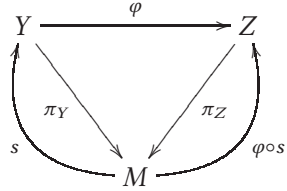
\[
\begin{xy}
(0,0)*+{\man}="X";(-15,20)*+{Y}="Y";
(15,20)*+{Z}="Z";{\ar@{->}_{\pi_Y}"Y";"X"};
{\ar@{->}^{\pi_Z}"Z";"X"};
{\ar@{->}^{\vp}"Y";"Z"};
{\ar@{->}@/^{2pc}/^{{s}}"X";"Y"};
{\ar@{->}@/_{2pc}/_{{\vp\pushf{}s}}"X";"Z"};\end{xy}
\]

Let $(Y,\pi,\man)$ be a fiber bundle and let $m\in\man$. We say that two (local) sections $s,u\in \Cone(Y)$ are \Emph{$1$-equivalent} at $m$ if $(Ts)_m=(Tu)_m$. Equivalently, $s$ and $u$ are $1$-equivalent if and only if they assume at $m$ the same values and the same first derivatives in some (hence, any) coordinate system.  The $1$-equivalence class at $m$ of a (local) section $s$ is denoted by $j^1_ms$. 
The \Emph{first jet bundle} of $\pi$ is the set 
\[
J^1\pi=\left\{ j^1_ms\,|\,m\in\man,\,\, \,s\,\,\text{is a local}\,\,\Cone\text{-section at }m    \right\}.
\]

In analogy with vertical bundles and $C^k$ sections, we adopt the following notation scheme: for a fiber bundle $(Y,\pi_Y,\man)$, we denote its jet bundle by $J^1Y$ in cases where $Y$ is the total space of a single bundle, or, in the case of repeated projections, $Y\to Z \to\cdots \to \man$, when the sections are  with respect to projection onto the manifold $\man$. 

The first jet bundle of $(Y,\pi,\man)$ is associated with the following projections:

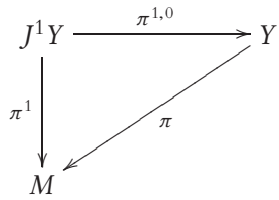
\[
\begin{xy}
(0,20)*+{J^1Y}="J1Y";(30,20)*+{Y}="Y";
(0,0)*+{\man}="X";{\ar@{->}^{\pi}"Y";"X"};
{\ar@{->}^{\pi^{1,0}}"J1Y";"Y"};
{\ar@{->}_{\pi^1}"J1Y";"X"};
\end{xy}
\]

Consider once again two fiber bundles $(Y,\pi,\man)$ and $(Z,\rho,\man)$ over
the same base manifold $\man$. Let $\vp:Y\to Z$ be a fiber bundle morphism.
The first \Emph{jet map} of $\vp$ is a fiber bundle morphism 
\[
j^1\vp:J^1Y\tto J^1Z,
\]
defined by 
\[
j^1\vp(j_{m}^1s)=j_{m}^1(\vp\pushf{}s),
\]
where for $m\in \man$, $s$ is a local section of $Y$ at $m$. 

\[
\begin{xy}
(0,0)*+{\man}="X";
(-15,15)*+{Y}="Y";(15,15)*+{Z}="VY";
(-15,30)*+{J^1Y}="J1Y";
(15,30)*+{J^1Z}="VJ1Y";
{\ar@{->}_{\pi}"Y";"X"};{\ar@{->}^{\rho}"VY";"X"};
{\ar@{->}^{\vp}"Y";"VY"};
{\ar@{->}_{\pi^{1,0}}"J1Y";"Y"};
{\ar@{->}^{\rho^{1,0}}"VJ1Y";"VY"};
{\ar@{->}^{j^1\vp}"J1Y";"VJ1Y"};
\end{xy}
\]


Let $Y$ and $Z$ be vector bundles over a manifold $\man$. We denote
by 
\[
\Hom(Y,Z)\simeq Y^*\otimes Z
\]
the vector bundle over $\man$ whose elements
at $m\in \man$ are linear maps $Y_{m}\to Z_{m}$. Using $L(Y,Z)$ to
designate the set of vector bundle morphisms $Y\to Z$, it is observed
that a vector bundle morphism $\xi\in L(Y,Z)$ can be identified with a section
of $\Hom(Y,Z)$; thus, a vector bundle morphism may be viewed as a
tensor field over $\man$. For
a section $\xi$ of $\Hom(Y,Z)$ and a section $\eta:\man\to Y$, we have
the section $\xi\pushf{}\eta:\man\to Z$, defined by $(\xi\pushf{}\eta)_{m}=\xi_{m}(\eta_{m})$.

Let $s\in\Omega^1(N)$ be a one-form and let $f:\man\to N$ be a mapping.   Then,
$f^*s$ is a section of $(f^*T^*N,f^*\tau_{N}^*,\man)$;
it is not a differential form. In contrast, we denote by $f^{\push}s$
the one-form over $\man$ defined by 
\begin{equation}
f^{\push}s(v)=f^*s(df(v)),
\label{eq:fsharp}
\end{equation}
for every $v\in T\man$. 

\begin{center}
\rule[0.5ex]{0.5\columnwidth}{0.5pt}
\par\end{center}

Throughout this paper we adopt the following terms and notation:


\begin{center}
\scriptsize  %
\begin{longtable}{llll}
\hline 
Objects  & Elements of  & Notation  \tabularnewline
\hline 
\hline 
Points in the field (base manifold) & $\base$ & $p$ \tabularnewline
Values of a field & $Y$ & $e$ \tabularnewline
Configurations  & $\Q\subset\Cone(Y)$  & $\bk,\bk_{1},\dots$  \tabularnewline
Virtual velocities  & $(T\Q,\tau_{\Q},\Q)$  & \tabularnewline
Velocities at $\bk$ & $T_{\bk}\Q\simeq\Cone(\bk^*VY)$  & $v,w,\dots$  \tabularnewline
Generalized forces  & $(T^*\Q,\tau_{\Q}^*,\Q)$  &    \tabularnewline
Forces  at $\bk$ & $T_{\bk}^*\Q\simeq(\Cone(\bk^*VY))^*$  & $f$   \tabularnewline
Body force densities  & $\vals B = \Hom(VY,\pi^*\LamD)$  &    \tabularnewline
Body force density field at $\bk$  & $\Czero(\bk^*\vals B)$  & $\ForceDensity$   \tabularnewline
Surface force densities  & $\vals T = \Hom(VY|_{\partial\base},(\pi|_{\partial\base})^*\ext^{d-1}T^*\partial\base)$  &    \tabularnewline
Surface force density field at $\bk$  & $\Czero(\bk_{\partial\base}^*\vals T)$  & $\SurfaceForceDensity$  \tabularnewline
Deformation jets  & $\strain\subset\Czero(J^1Y)$  & $\xi$  \tabularnewline
Velocity jets  & $(T\strain,\tau_{\strain},\strain)$  &   \tabularnewline
Velocity jets at $\xi$  & $T_{\xi}\strain\simeq\Czero(\xi^*VJ^1Y)$  & $\eta$   \tabularnewline
Variational stresses  & $(T^*\strain,\tau_{\strain}^*,\strain)$  &   \tabularnewline
Variational stresses at $\xi$ & $T_{\xi}^*\strain\simeq(\Czero(\xi^*VJ^1Y))^*$  & $\sigma$   \tabularnewline
Variational stress densities & $\stds = \Hom(VJ^1Y,(\pi^1)^*\LamD)$  &    \tabularnewline
Variational stress density fields at $\xi$ & $C^0(\xi^*\stds)$  & 
$\StressDensity$  \tabularnewline
Traction stress densities & $\trstds = \Hom(VY,\pi^*\ext^{d-1}T^*\base)$ & \tabularnewline
Traction stress density fields at $\bk$ & $\Czero(\bk^*\trstds)$ & $\trst$ \tabularnewline
Loadings  & $\Czero(\tau_{\Q}^*)$  & $F$  \tabularnewline
Body loading densities  & $\Czero(\Hom(VY,\pi^*\LamD))$  & $\LoadingDensity$  \tabularnewline
Surface loading densities  & $\Czero(\Hom(VY\resto{\partial\base},(\pi\resto{\partial\base})^*\ext^{d-1}T^*\partial\base))$  & $\SurfaceLoadingDensity$  \tabularnewline
Loading potentials  & $\Cone(\Q)$  & $W$  \tabularnewline
Body loading potential densities  & $\Cone(\pi^*\LamD)$  & $w_\base$  \tabularnewline
Boundary loading potential densities  & $\Cone((\pi\resto{\partial\base})^*\ext^{d-1}T^*\partial\base)$  & $w_{\partial\base}$  \tabularnewline
Constitutive relations  & $\Czero(\tau_{\strain}^*)$  & $\Psi$  \tabularnewline
Constitutive densities & $\Czero(\Hom(VJ^1Y,(\pi^1)^*\LamD))$  & $\ConstitutiveDensity$   \tabularnewline
Elastic energy  & $\Cone(\strain)$  & $U$  \tabularnewline
Elastic energy density  & $C^2((\pi^1)^*\LamD)$  & $\calL$  \tabularnewline
\hline
\end{longtable}
\par\end{center}


\section{Configurations, Velocities and Forces}

\subsection{The manifold of configurations}

In the global approach to mechanics, a system is characterized by
its configuration space. 

The fundamental object in a classical field theory is a fiber bundle, in which the various fields assume their
values. The $d$-dimensional base manifold $\base$ typically represents
space-time in modern physical field theories and a body manifold in continuum mechanics. 
For $p\in\base$,
the $m$-dimensional fiber
$Y_p$ represents the values that the field
may assume at $p$. Thus, a field theory is characterized
by a particular fiber bundle. 

\begin{defn}
\label{def:manifold_of_configurations} 
Let $(Y,\pi,\base)$ be a fiber bundle, where the base manifold $\base$ is assumed to be compact, oriented and possibly having a boundary. 
Consider  the Banach manifold
\cite{Palais68} of sections
$\Cone(Y)$. The \Emph{manifold
of configurations}, $\Q$, is an open subset of $\Cone(Y)$.
\end{defn}

Since $\Q$ is open in $\Cone(Y)$, it inherits its Banach manifold structure; moreover, for every $\bk\in\Q$, $T_\bk\Q=T_\bk \Cone(Y)$.

\begin{cmnt}
 A basic example of a manifold of configurations is the case where $Y$ is a trivial bundle. In the
Lagrangian approach to continuum mechanics, a \Emph{body} is modeled
as a smooth, compact, $d$-dimensional differentiable manifold, $\base$. The ambient \Emph{space} is modeled as a smooth $m$-dimensional differentiable manifold without
boundary, $S$.
The space of configurations is the space of $\Cone$-embeddings $\base\to S$, which
can be given the structure of a smooth Banach manifold \cite{Palais68,michor80}.
We may also view such maps as sections, $\Cone(Y)$, where 
\[
Y=\base\times S
\]
is a trivial bundle over $\base$.
\end{cmnt}

A construction of the manifold $\Cone(Y)$ consistent with the Whitney
$\Cone$-topology \cite{michor80} can be found in Palais \cite{Palais68}.
The construction may be roughly described as follows. Let $\bk\in\Q$
and let $\Cone(\bk^*VY)$ be the Banachable space
of sections $\base\to VY$ along $\bk$. That is, $w\in\Cone(\bk^*VY)$
satisfies for $p\in\base$
\[
w_p\in(\bk^*VY)_p = (VY)_{\bk_p}.
\]
A local chart for $\Q$ in a neighborhood of $\bk$ is a map 
\[
\chi_\bk:\Cone(\bk^*VY)\to\Q,
\]
given by 
\[
(\chi_\bk(w))_p = \exp^{Y}(w_p),
\]
where $\exp^{Y}:VY\to Y$ is the exponential map on $Y$ induced by some arbitrarily chosen spray, consistent
with the fiber structure. Namely, for $e\in Y$,
\[
\exp^{Y}:(VY)_e\to Y_{\pi(e)}.
\]
Strictly speaking, $\exp^Y$ is only defined on a neighborhood of the zero section of $VY$. However, one can always reparametrize $\exp^Y$ to be defined globally on $VY$.
The above construction holds only for the case of a compact base manifold $\base$. See Remark \ref{rem:nonc} below for the case of non-compact bases.

\InCoord{ 
Throughout this paper, we complement the covariant, coordinate-free
formulation with its corresponding local coordinate representation.
We will use a typical local coordinate system
\[
X=(X^1,\dots,X^{d}):p\in\base\lmt X(p)\in\R^{d}
\]
for the base manifold $\base$, and a local coordinate system
\[
y=(X,x)=(X^1,\dots,X^{d},x^1,\dots,x^m):e\in Y \lmt y(e)\in\R^{d}\times\R^m
\]
for the fiber bundle $Y$. The components of $X$ for a given chart
will be denoted with Greek indexes, \eg, $X^\alpha$; the components
of $x$ will be denoted with Roman indexes, \eg, $x^i$. Note the
abuse of notation where $X$ is both a function on $\base$ and a function
on $Y$; this type of abuse will recur in several instances below.

The coordinate system $y$ is assumed to be adapted to the bundle
structure: for every $e\in Y$, \emph{i.e.},
\[
X(e)=X(\pi(e)).
\]
}

\subsection{Generalized velocities}

\begin{defn}
\label{def:bundle_of_velocities} The bundle $(T\Q,\tau_{\Q},\Q)$
tangent to the manifold of configurations is termed the \Emph{bundle
of generalized velocities}, or the \Emph{bundle of virtual displacements}. 
\end{defn}

Following Lang \cite[p.~26]{Lang72}, we define the tangent space of an 
infinite-dimensional manifold in the following way.
Let $\bk_1$ and $\bk_2$ be two
configurations, and let 
\[
\chi_1:\Cone(\bk_1^*VY)\to\Q\Textand\chi_2:\Cone(\bk_2^*VY)\to\Q
\]
be coordinate systems at $\bk_1$ and $\bk_2$, respectively, whose images overlap. We will keep the simple notation $ \chi_2^{-1} $ for the restriction to the overlap.
Let 
\[
\bk=\chi_1(w_1)=\chi_2(w_2).
\]
The triples 
\[
(\bk,\chi_1,v_1)\Textand(\bk,\chi_2,v_2),
\]
where $v_1\in\Cone(\bk_1^*VY)$ and $v_2\in\Cone(\bk_2^*VY)$,
are considered equivalent if 
\[
v_2=D_{w_1}(\chi_2^{-1}\circ\chi_1)(v_1).
\]
The collection of all such equivalence classes for a fixed $\bk$
forms the vector space $T_\bk\Q$. Given an exponential map on $Y$,
the canonical representative of an element of $T_\bk\Q$ is of the
form 
\[
(\bk,\chi_\bk,v),\qquad v\in\Cone(\bk^*VY),
\]
that is, the vector space of virtual velocities at $\bk$ can be identified with the space of velocity fields at $\bk$,
\begin{equation}
T_\bk\Q\simeq  \Cone(\bk^*VY).
\label{eq:canon2}
\end{equation}

A velocity field at $\bk$, $v\in \Cone(\bk^*VY)$ can be identified with
a path $\gamma:I\to\Q$ in a canonical way, 
\[
\gamma(t)=\chi_\bk(tv),
\]
\ie, 
\[
(\gamma(t))_p=\exp^{Y}(t\,v_p)\in Y_p.
\]

The bundle of velocities $T\Q$ is the union of the tangent spaces $T_\bk\Q$
with the standard smooth structure.
As a set,   
$T\Q$ consists of sections of $VY$ viewed as a fiber bundle over $\base$. The tangent space at $\bk\in\Q$ consists of those sections
whose projection onto $Y$ coincides with $\bk$.
In
the physics context, 
an element of $T_\bk\Q$ is interpreted as a fiber-wise variation
of $\bk$.

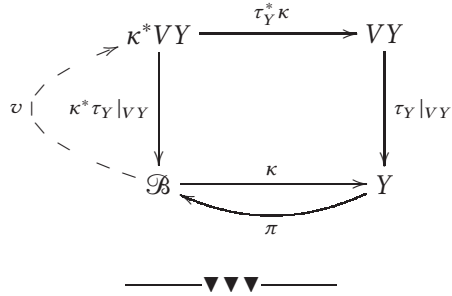
\[
\begin{xy}
(0,0)*+{\base}="B";(30,0)*+{Y}="Y";
(0,20)*+{\bk^*VY}="sVY";
(30,20)*+{VY}="VY";{\ar@{->}^{\bk}"B";"Y"};
{\ar@{->}_{\bk^*\tau_{Y}|_{VY}}"sVY";"B"};
{\ar@{->}^{\tau_{Y}|_{VY}}"VY";"Y"};
{\ar@{->}^{\tau_Y^*\bk}"sVY";"VY"};
{\ar@{->}@/^{1pc}/^{\pi}"Y";"B"};
{\ar@{-->}@/^{4pc}/^{{v}}"B";"sVY"};\end{xy}
\]

\InCoord{
The local coordinate system $X$ on $ \base$ induces a smooth local frame for $T\base$,
\[
\{\parB_\alpha:=\frac{\partial}{\partial X^\alpha}~:~\alpha=1,\dots,d\},
\]
defined by the paths 
\[
\parB_\alpha|_p=[t\mapsto X^{-1}(X(p)+t\,e_\alpha)],\qquad p\in\base,
\]
where $\{e_\alpha\}$ is the standard basis of $\R^{d}$. The action of this frame on a function
$f\in \scrD(\base)$ is
\[
(\parB_\alpha f)(p) = D_{X(p)}^\alpha(f\circ X^{-1}).
\]

Likewise, the local coordinate system $y=(X,x):Y\to\R^n$ induces a smooth local frame for $TY$,
\[
\{\parY_i,\parY_\alpha ~:~ \alpha=1,\dots,d,\,\,i=1,\dots,m\},
\]
defined by the paths 
\[
\begin{split}
\parY_\alpha|_e &= [t\mapsto y^{-1}(X(e)+t\,e_\alpha,x(e))]\\
\parY_i|_e & =[t\mapsto y^{-1}(X(e),x(e)+t\,e_i)],\qquad e\in Y.
\end{split}
\]
In terms of derivations, for $f\in \scrD(Y)$ and $e\in Y$,
\[
\begin{aligned}
(\parY_\alpha f)(e) & =D_{y(e)}^\alpha(f\circ y^{-1}),\qquad\alpha=1,\dots,d\\
(\parY_if)(e) & =D_{y(e)}^i(f\circ y^{-1}),\qquad i=1,\dots,m.
\end{aligned}
\]

Since the coordinate systems are adapted, 
\begin{equation}
d\pi\circ\parY_\alpha=\pi^*\,\parB_\alpha\Textand d\pi\circ\parY_i=0.
\label{eq:useful1}
\end{equation}
The vertical bundle $VY$ is spanned locally by the frame field
$\parY_i$; a local coordinate system for $VY$ is 
\[
(X,x,\dot{x}): v\in VY\lmt(X,x,\dot{x})(v)\in\R^{d}\times\R^m\times\R^m,
\]
where 
\[
X^\alpha\brk{v^i\,\parY_i|_e}=X^\alpha(e),\qquad x^j\brk{v^i\,\parY_i|_e}=x^j(e)
\textand
\dot{x}^j\brk{v^i\,\parY_i|_e}=v^j.
\]
A generalized velocity (or virtual displacement) field at $\bk$ has a local representation
\[
v=v^i\,\bk^*\parY_i,
\]
where  $v^i\in \Cone(\base)$.

We denote by $\dxB^\alpha$ and $\dxY^\alpha$, $\dxY^i$ the
co-frames dual to $\parB_\alpha$, $\parY_\alpha$ and $\parY_i$.
We have 
\begin{equation}
\dxY^\alpha=\pi^{\push}\dxB^\alpha,\label{eq:useful2}
\end{equation}
because by \eqref{eq:fsharp} and \eqref{eq:useful1} 
\[
\pi^{\push}\dxB^\alpha(\parY_{\beta})=\pi^*\dxB^\alpha(d\pi\circ\parY_{\beta})=\pi^*\dxB^\alpha(\pi^*\,\parB_{\beta})=\delta_{\beta}^\alpha
\]
and 
\[
\pi^{\push}\dxB^\alpha(\parY_i)=\pi^*\dxB^\alpha(d\pi\circ\parY_i)=0.
\]
} 

\subsection{Generalized forces}

\begin{defn}
\label{def:bundle_of_forces} 
The \Emph{bundle of forces} is the vector bundle 
\[
(T^*\Q,\tau^*_{\Q},\Q)
\]
dual to the vector bundle of velocities. 
A \Emph{generalized force} at $\bk$ is an element
\[
f\in T_\bk^*\Q.
\]
\end{defn}

That is, for every configuration $\bk\in\Q$, the vector space $T_\bk^*\Q$
of forces at $\bk$ is the dual of the vector space $T_\bk\Q$ of
velocities at $\bk$. The action of a force $f$ at $\bk$ on a velocity
$v$ at $\bk$ yields a real number, $f(v)$, termed the \Emph{virtual
power}, or \Emph{virtual work} that $f$ expends on $v$.

By the isomorphism \eqref{eq:canon2}, 
\begin{equation}\label{iso10}
T_\bk^*\Q\simeq(\Cone(\bk^*VY))^*.
\end{equation}

Generally, forces may be represented by a collection of
measures (see Subsection \ref{sec5p1}) and therefore cannot be assigned values at points. The remaining part of this section considers forces that
can be represented by more regular fields.

\begin{defn}
The bundle of \Emph{body force densities} is 
\[
\vals B :=\Hom(VY,\pi^*\ext^{d}T^*\base).
\]
It is a vector bundle over $Y$, with projection which we denote by $\pi_{\vals B}:\vals B\to Y$.
For $\bk\in\Q$,
\[
\bk^*\vals B =
\Hom(\bk^*VY,\ext^{d}T^*\base)
\]
is the bundle of body force densities along $\bk$; it is a vector bundle over $\base$.

The bundle of \Emph{surface force densities} is 
\[
\vals T:=\Hom(VY|_{\partial\base},(\pi|_{\partial\base})^*\ext^{d-1}T^*\partial\base).
\]
It is a vector bundle over $Y|_{\partial\base}$, with projection which we denote by $\pi_{\vals T}$.
For $\bk\in\Q$, and $\bk_{\partial\base}:=\bk\resto{\partial\base}$,
\[
\bk_{\partial\base}^*\vals T = 
 \Hom(\bk^*VY|_{\partial\base},\ext^{d-1}T^*\partial\base)
\]
is the bundle of surface force densities along $\bk_{\partial\base}$; it is a vector bundle over $\partial\base$. 
\end{defn}

With these definitions, we can define the notion of a continuous force:

\begin{defn}
\label{def:smooth_force} 
A force $f$ at $\bk$ is termed \Emph{continuous}
if there exists a \Emph{body force density field} along $\bk$,
\[
\ForceDensity\in \Czero(\bk^*\vals B)  \simeq \Czero(\Hom(\bk^*VY,\LamD)),
\]
and a \Emph{surface force density field} along $\bk$, 
\[
\SurfaceForceDensity\in \Czero(\bk_{\partial\base}^*\vals T) \simeq \Czero(\Hom(\bk^*VY|_{\partial\base},\ext^{d-1}T^*\partial\base)),
\]
such that for every velocity $v$ at $\bk$, the virtual power that
$f$ expends on $v$ is given by 
\begin{equation}
f(v)=\int_{\base}\ForceDensity\pushf v+\int_{\partial\base}\SurfaceForceDensity\pushf v|_{\partial\base}.\label{eq:smooth_force}
\end{equation}
Note that on the left-hand side of \eqref{eq:smooth_force}, 
$v$ is viewed as an element of $T_\bk\Q$, whereas on the right-hand side,
$v$ is viewed as a velocity field at $\bk$, i.e., an element of $\Cone(\bk^*VY)$.
\end{defn}

\begin{remark}[\textbf{Non-compact base manifolds}]
\label{rem:nonc}
If the base manifold $\B$ is not compact, the image of a section $s\in C^1(Y)$ is never compact, hence, it is not possible to endow $C^1(E)$ (or any other reasonable class of sections) with a Banach manifold structure (see discussion in the introduction of \cite{TaP01}). However, as shown
in \cite{michor80} and \cite{Krm97} for smooth sections, the space $C^\infty(Y)$ of $C^\infty$-sections $\B\to Y$ can be given a structure of a smooth manifold modeled on a locally convex topological vector space. The tangent space at a configuration $\bk\in C^\infty(Y)$ may be identified with the space $C_c^\infty(\bk^*VY)$ of differentiable sections with compact supports, equipped with the inductive limit topology 
\[
C_c^\infty(\bk^*VY)=\lim_K C_K^\infty(\bk^*VY)
\]
where $K$ runs through all compact sets in $\B$ and $C_K^\infty(\bk^*VY)$ has the topology of uniform convergence of all derivatives \cite[Theorem 42.1]{Krm97}. Thus, generalized velocities are represented by sections of $\bk^*VY $ having compact supports so that forces may be viewed as tensor valued currents or generalized sections (see \cite{GuillAndStrenberg,Segev_e_d_2016}).
%
An analogous construction can be applied to $C^1$-sections.

%
%

\end{remark}

\section{Deformation Jets and Velocity Jets}

\subsection{The manifold of deformation jets}

In this section, we consider what is termed in \cite{Segev1981} the local model\textemdash a
notion of configuration that encodes also information about local
deformations. In a first grade theory, this additional information
is reflected by conceivable values of the first derivative of the configuration
 at the various points. These values of the derivatives need
not be compatible with a particular configuration $\bk$. 
For short, we refer to such fields as deformation jets (referred to in \cite{Segev1981}, for the case of a trivial
bundle, as local configurations after \cite[pp. 51--52]{Truesdell1965}). 
The natural geometric construct for encoding this
information is the first jet bundle $(J^1Y,\pi^1,\base)$. 

\begin{defn}
\label{def:manifold_of_strains} 
The \Emph{manifold of deformation jets} which we denote by $\strain$, is some open subset
of $\Czero(J^1Y)$---the space of $\Czero$-sections $\base\to J^1Y$ containing the image of $j^1:\Q\to C^0(J^1Y)$. That is 
\[
j^1(\Q)\subset \strain\subset C^0(J^1Y).
\]

 A deformation jet which is a jet extension $j^1\bk\in\strain$ of a configuration $\bk\in\Q$ is termed \Emph{compatible} or \Emph{holonomic}. 
\end{defn}

The manifold structure of $\strain$ is analogous to the manifold
structure of $\Q$. 
 Denote the vertical bundle of $J^1Y$ by 
\[
VJ^1Y=\ker T\pi^1\subset TJ^1Y.
\]

\[
\begin{xy}
(20,0)*+{\base}="X";
(0,0)*+{J^1Y}="Y";
(20,20)*+{T\base}="TX";
(0,20)*+{TJ^1Y}="TY";
(-20,20)*+{VJ^1Y}="VY";
{\ar@{->}_{\pi^1}"Y";"X"};
{\ar@{->}^{T\pi^1}"TY";"TX"};
{\ar@{->}^{\tau_{\base}}"TX";"X"};
{\ar@{->}_{\tau_{J^1Y}}"TY";"Y"};
{\ar@{->}_{\tau_{J^1Y}|_{VJ^1Y}}"VY";"Y"};
{\ar@{^{(}->}^{\text{incl.}}"VY";"TY"};
\end{xy}
\]

For $\xi\in\strain$, the modeling space for a chart in a neighborhood
of $\xi$ is the Banachable space $\Czero(\xi^*VJ^1Y)$
of sections $\base\to VJ^1Y$ along $\xi$; that is, $\xi'$ in that
neighborhood of $\xi$ is represented by an element $\eta\in\Czero(\xi^*VJ^1Y)$
satisfying 
\[
\eta_p\in(\xi^*VJ^1Y)_p\simeq(VJ^1Y)_{\xi_p}\subset T_{\xi_p}J^1Y.
\]

\[
\begin{xy}
(0,0)*+{\base}="X";
(30,0)*+{J^1Y}="J1Y";
(30,20)*+{VJ^1Y}="VJY";
(0,20)*+{\xi^*VJ^1Y}="xiVJY";
{\ar@{->}_{\xi}"X";"J1Y"};
{\ar@{->}^{\tau_{J^1Y}|_{VJ^1Y}}"VJY";"J1Y"};
{\ar@{->}^{(\tau_{J^1Y}\resto{VJ^1Y})^*\xi}"xiVJY";"VJY"};
{\ar@{->}_{\xi^*(\tau_{J^1Y}|_{VJ^1Y})}"xiVJY";"X"};
{\ar@{-->}@/^{5pc}/^{{\eta}}"X";"xiVJY"};
\end{xy}
\]

\InCoord{ 
A local coordinate system for $J^1Y$ \cite[p. 96]{Saunders} is 
\[
z=(y,x')=(X,x,x'):j_p^1\bk\in J^1Y\lmt z(j_p^1\bk)\in\R^{d}\times\R^{m}\times\R^{d\times m},
\]
such that 
\[
y(j_p^1\bk)=y(\bk_p)\Textand{x'}_\alpha^i(j_p^1\bk)=\parB_\alpha(x^i\circ\bk)(p).
\]
} 

\subsection{The bundle of velocity jets}

The analog of a generalized velocity for the case of the bundle
of deformation jets is defined as follows:

\begin{defn}
\label{def:bundle_of_strain_rates} The vector bundle $T\strain$
tangent to the manifold of deformation jets $\strain$ will be termed the \Emph{bundle
of velocity jets}. In \cite{Segev1981} it is
referred to as the bundle of local virtual displacements
(again, in the context of a trivial bundle). 
\end{defn}
The construction of $T\strain$ is analogous to the construction of
the bundle of velocities $T\Q$. In particular, for $\xi\in\strain$,
\begin{equation}
T_{\xi}\strain\simeq\Czero(\xi^*VJ^1Y),\label{eq:iso4}
\end{equation}
which is the space of sections $\base\to VJ^1Y$ along $\xi$. 

As a set,   
$T\strain$ consists of sections of $VJ^1Y$ viewed as a fiber bundle over $\base$. The tangent space at $\xi\in\strain$ consists of those sections
whose projection onto $J^1Y$ coincides with $\xi$.
In
the physics context, 
an element of $T_\xi\strain$ is interpreted as a fiber-wise variation
of $\xi$.

For later use, we will need the following well-known isomorphism:
\begin{prop}
\label{prop:isoSaunders} $VJ^1Y$ is isomorphic to $J^1VY$,
\[
(VJ^1Y,\pi^1\circ\tau_{J^1Y}|_{VJ^1Y},\base)\simeq(J^1VY,(\pi\circ\tau_{Y}|_{VY})^1,\base).
\]
We will denote the isomorphism by 
\[
K:VJ^1Y\tto J^1VY.
\]
\end{prop}
\[ \begin{xy} (0,0)*+{\base} = "X";%
(0,15)*+{Y} = "Y";%
(-30,30)*+{VY} = "VY";%
(30,30)*+{J^1Y} = "J1Y";%
(-30,50)*+{J^1VY} = "all1";%
(30,50)*+{VJ^1Y} = "all2";%
{\ar@{->}_{\pi} "Y"; "X"};%
{\ar@{->}_{\tau_Y|_{VY}} "VY"; "Y"};%
{\ar@{->}^{\pi^{1,0}} "J1Y"; "Y"};%
{\ar@{->}^{(\pi\circ\tau_Y|_{VY})^{1,0}} "all1"; "VY"};%
{\ar@{->}_{\tau_{J^1Y}|_{VJ^1Y}} "all2"; "J1Y"};%
{\ar@{->}_{K} "all2"; "all1"};%
{\ar@{->}^(0.7){T\pi^{1,0}|_{VJ^1Y}} "all2"; "VY"};%
{\ar@{->}_(0.7){J^1(\tau_Y|_{VY})} "all1"; "J1Y"};%
\end{xy} \]

\begin{proof}
See \cite[Chapter~17]{Palais68} and \cite[Theorem~4.4.1]{Saunders}. 
\end{proof}

\begin{lem}
\label{lem:jetsPBs}
Let $\jk\in\strain$ be a holonomic deformation jet. Then, the isomorphism $K$
induces an isomorphism  of vector bundles over $\base$,
\[
(\jk)^*K: (\jk)^*VJ^1Y\tto J^1(\bk^*VY).
\]
\end{lem}
\begin{proof}
See Theorem 17.1 in \cite[p. 82]{Palais68}.
\end{proof}

\InCoord{
A local coordinate system for $VJ^1Y$ is
\[
(z,\dot{x},\dot{x}')=(X,x,x',\dot{x},\dot{x}'):
(j_p^1s)^\cdot \in VJ^1Y\lmt
(z,\dot{x},\dot{x}')((j_p^1s)^\cdot)\in
\R^{d}\times\R^{m}\times\R^{d\times m}\times\R^{m}\times\R^{d\times m}.
\]
This local system of coordinates is adapted in the sense that for $[t\mapsto j_p^1s(t)]\in VJ^1Y$,
where $s(0)=\bk$, 
\[
z([t\mapsto j_p^1s(t)])=z(j_p^1\bk).
\]
For $j_p^1\bk\in J^1Y$ 
\[
\begin{aligned}\parJY_\alpha|_{j_p^1\bk} & =[t\mapsto z^{-1}(X(j_p^1\bk)+t\,e_\alpha,x(j_p^1\bk),x'(j_p^1\bk))]\\
\parJY_i|_{j_p^1\bk} & =[t\mapsto z^{-1}(X(j_p^1\bk),x(j_p^1\bk)+t\,e_i,x'(j_p^1\bk))]\\
(\parJY)_i^\alpha|_{j_p^1\bk} & =[t\mapsto z^{-1}(X(j_p^1\bk),x(j_p^1\bk),x'(j_p^1\bk)+t\,e_i^\alpha)],
\end{aligned}
\]
we have 
\[
z\brk{\parJY_{\beta}|_{j_p^1\bk}}=z(j_p^1\bk),\qquad z\brk{\parJY_j|_{j_p^1\bk}}=z(j_p^1\bk),\qquad z\brk{(\parJY)_j^{\beta}|_{j_p^1\bk}}=z(j_p^1\bk).
\]
Moreover, 
\[
\dot{x}^i\brk{\parJY_{\beta}|_{j_p^1\bk}}=0,\qquad\dot{x}^i\brk{\parJY_j|_{j_p^1\bk}}=\delta_j^i,\qquad\dot{x}^i\brk{(\parJY)_j^{\beta}|_{j_p^1\bk}}=0,
\]
and 
\[
{\dot{x}'}\,_\alpha^i\brk{\parJY_{\beta}|_{j_p^1\bk}}=0,\qquad{\dot{x}'}\,_\alpha^i\brk{\parJY_j|_{j_p^1\bk}}=0,\qquad{\dot{x}'}\,_\alpha^i\brk{(\parJY)_j^{\beta}|_{j_p^1\bk}}=\delta_j^i\delta_\alpha^{\beta}.
\]

A local coordinate system for $J^1VY$
is 
\[
(X,x,\dot{x},x',\dot{x}'): j_p^1v\in J^1VY\lmt(X,x,\dot{x},x',\dot{x}')(j_p^1v)\in
\R^{d}\times\R^{m}\times\R^{m}\times\R^{d\times m}\times\R^{d\times m}.
\]
For a local section $v=v^i\,\bk^*\parY_i$ of $\bk^*VY$,
\[
X^\alpha(j_p^1v)=X^\alpha(v_p)\qquad x^i(j_p^1v)=x^i(v_p)\qquad\dot{x}^i(j_p^1v)=v^i(p)
\]
\[
{x'}_\alpha^i(j_p^1v)={x'}_\alpha^i(j_p^1\bk)\Textand(\dot{x}')_\alpha^i(j_p^1v)=(\parB_\alpha v^i)(p).
\]

The isomorphism $VJ^1Y\simeq J^1VY$
is represented by: 
\[
K\brk{v^i(p)\,\parJY_i|_{j_p^1\bk}+(\parB_\alpha v^i)(p)\,(\parJY)_i^\alpha|_{j_p^1\bk}}=j_p^1(v^i\,\bk^*\parY_i),
\]
In other words, The local representative $\rep K$ of $K$ is given by
\[
\rep K(X,x,x',\dot{x},\dot{x}')=(X,x,\dot{x},x',\dot{x}').
\]
Note that 
\begin{equation}
d\pi^{1,0}\circ\parJY_\alpha=(\pi^{1,0})^*\parY_\alpha\qquad d\pi^{1,0}\circ\parJY_i=(\pi^{1,0})^*\parY_i\textand d\pi^{1,0}\circ(\parJY)_i^\alpha=0.
\label{eq:useful3}
\end{equation}
We denote by $\dxJY^\alpha$, $\dxJY^i$ and $(\dxJY)_\alpha^i$
the corresponding co-frames: 
\begin{equation}
\dxJY^\alpha=(\pi^{1,0})^{\push}\dxY^\alpha=(\pi^1)^{\push}\dxB^\alpha\Textand\dxJY^i=(\pi^{1,0})^{\push}\dxY^i.
\label{eq:useful4}
\end{equation}
} 

\subsection{Compatibility and jet prolongation of velocity fields}

Evidently, special attention should be given to configuration jets induced
as jets of configurations. The analogous
situation applies to generalized velocity fields. In this section, we consider compatible
configuration jets and compatible velocity jets.

The jet prolongation mapping $j^1:\Q\to \strain$ is an injection,
where we omit the indication that $j^1$ needs to be restricted
first from $\Cone(Y)$ to $\Q$. 
Its differential
is a vector bundle morphism (see diagram below)
\[
dj^1:T\Q\tto(j^1)^*T\strain,
\]
mapping velocities at $\bk$ into velocity jets at $\jk$, 
\[
(dj^1)_{\bk} : T_{\bk}\Q\tto   ((j^1)^*T\strain)_\bk = T_{\jk}\strain.
\]
\[
\begin{xy}
(0,0)*+{\Q} = "Q";
(40,0)*+{\strain} = "S";
(40,15)*+{T\strain} ="TS";
(-10,15)*+{T\Q} = "TQ";
(15,15)*+{(j^1)^*T\strain} ="jTS";
{\ar@{->}_{j^1} "Q"; "S"};
{\ar@{->}_{\tau_\Q} "TQ"; "Q"};
{\ar@{->}^{(j^1)^*\tau_\strain} "jTS"; "Q"};
{\ar@{->}^{\tau_\strain} "TS"; "S"};
{\ar@{->}@/^{2pc}/^{Tj^1} "TQ"; "TS"};
{\ar@{->}^{dj^1} "TQ"; "jTS"};
{\ar@{->}^{{\tau_\strain^*}j^1} "jTS"; "TS"};
\end{xy}
\]

Since $T_{\bk}\Q\simeq\Cone(\bk^*VY)$ and $T_{\jk}\strain\simeq\Czero((\jk)^*VJ^1Y)$,
the differential of $j^1$ at $\bk$ can also be viewed as a linear
map 
\[
(dj^1)_{\bk}:\Cone(\bk^*VY)\tto\Czero((\jk)^*VJ^1Y),
\]
mapping velocity fields at $\bk$ into velocity jet fields at $j^1\bk$.

\begin{prop}
\label{prop:dj1} The differential of $j^1$ can be factored into
the action of $j^1$ and a vector bundle morphism: for $v\in T_{\bk}\Q\simeq\Cone(\bk^*VY)$,
\[
(dj^1)_{\bk}(v)=(\jk)^*K^{-1}\circ j^1v.
\]
In other words, $(dj^1)_\bk$ is represented by 
\[
j^1: \Cone(\bk^*VY) \to \Czero(J^1(\bk^*VY)) \simeq \Czero((j^1\bk)^*VJ^1Y),
\]
where the last isomorphism follows from Lemma~\ref{lem:jetsPBs}.
\end{prop}

\begin{proof}
For $v\in T_{\bk}\Q$, let $s:I\to\Q$ be a path of configurations
satisfying 
\[
s(0)=\bk\Textand\dot{s}(0)=v.
\]
By the definition of the differential via its action on curves, 
\[
(dj^1)_{\bk}(v)=[t\mapsto j^1s(t)]\in T_{\jk}\strain.
\]
We now view $v$ as a section $\Cone(\bk^*VY)$, \ie,
as a map 
\[
p\mapsto[t\mapsto s_p(t)].
\]
Then, 
\[
j^1v=j^1[t\mapsto s(t)]\in\Czero(J^1(\bk^*VY)),
\]
and 
\[
(\jk)^*K^{-1}\circ j^1v=[t\mapsto j^1s(t)]\in\Czero((\jk)^*VJ^1Y).
\]
\end{proof}
Since $(dj^1)_{\bk}(v)$ can be identified with $j^1v$, we will
use the shorter notation $j^1v$ rather than $(dj^1)_{\bk}(v)$
to denote the corresponding element of $T_{j^1\bk}\strain$, and treat 
\[
j^1:\Cone(\bk^*VY)\tto\Czero(J^1(\bk^*VY))
\]
as a representative of $(dj^1)_{\bk}$ .

\InCoord{
Using local coordinate frames, a velocity field $v$ at
$\bk$ may be represented locally in the form, 
\[
v=v^i\,\bk^*\parY_i,
\]
where $v^i$ are differentiable functions defined on the domain
of a chart. Its jet prolongation is the velocity jet field at $\jk$, represented
locally as
\begin{equation}
j^1v=v^i\,(\jk)^*\parJY_i+(\parB_\alpha v^i)\,(\jk)^*(\parJY)_i^\alpha.
\label{eq:j1v}
\end{equation}
} 


\section{Stresses}\label{sec:Stress}

This section introduces the stress object as a tensor valued measure that represents a force functional, non-uniquely. Particular attention is given to stresses measures that are continuous relative to volume measures on the manifold $ \base $.

\subsection{Variational stresses}
\label{sec5p1}

\begin{defn}
\label{def:bundle_of_stresses} The bundle $(T^*\strain,\tau_{\strain}^*,\strain)$
dual to the bundle of velocity jets $T\strain$ is termed the \Emph{bundle
of variational stresses}. Given a deformation jet $\xi\in\strain$, an element $\sigma\in T^*_{\xi}\strain$ is referred to as a \Emph{variational stress} at $\xi$.
\end{defn}

For every deformation jet $\xi\in\strain$, the vector space $T_{\xi}^*\strain$
is the dual of the vector space of velocity jets at $\xi$, $T_{\xi}\strain$. By the isomorphism
\eqref{eq:iso4}, 
\begin{equation}
T_{\xi}^*\strain\simeq(\Czero(\xi^*VJ^1Y))^*.
\end{equation}
Let $\bk\in\Q$ be given. The map 
\[
j^1: C^1(\bk^*VY)\tto C^0((j^1\bk)^*VJ^1Y)
\]
is an embedding. It follows from the Hahn-Banach theorem that its
dual,
\[
(j^1)^*:  (C^0((j^1\bk)^*VJ^1Y))^* \tto (C^1(\bk^*VY))^*
\]
is surjective; to every force at $\bk$, $f\in (C^1(\bk^*VY))^*$,
there corresponds a (non-unique) variational stress at $\jk$, $\sigma\in (C^0((j^1\bk)^*VJ^1Y))^*$,
such that 
\begin{equation}
f=(j^1)^*\sigma.\label{eq:Gen_Eq}
\end{equation}
That is, for every $v\in T_{\bk}\Q$, 
\begin{equation}
f(v)=\sigma(j^1v).
\label{eq:force_rep_by_stress}
\end{equation}
Equation \eqref{eq:force_rep_by_stress} is a generalization of the
\Emph{principle of virtual work} in continuum mechanics, and Equation
(\ref{eq:Gen_Eq}) is the corresponding generalization of the \Emph{equilibrium equation}.

It should be noted that the (generalized) equilibrium equation is
merely a representation theorem; it is not a law of physics. Note
also that the well-known \Emph{static indeterminacy}---the non-uniqueness of the stress representing a given force---is reflected by the non-injectivity of $(j^1)^*$, which in turn, follows from the fact that $ j^1 $ is not surjective.

Let $\bk\in \Q$, hence $j^1\bk\in \strain$. By the Riesz representation theorem, the space of continuous linear functionals on  $\Czero$-sections, 
\[
T_{j^1\bk}^*\strain \simeq (\Czero((j^1\bk)^*VJ^1Y))^*,
\]
coincides with the space of Radon measures valued in the dual vector bundle $((j^1\bk)^*VJ^1Y)^*$.
 Locally,  a variational stress $\sigma\in T^*_{j^1\bk}\strain$ is represented by a collection of Radon measures 
\[
\{\mu_i,\,\mu_i^\alpha ~:~  1\leq i\leq m,\,\,  1\le \alpha \le d\}
\] 
 so that in case $ v $ or $ \sigma $ are supported in the domain of a single chart,
\begin{equation}
\label{stressrepmeasures}
\sigma(j^1v)=\int_\base v^i \,d\mu_i+\int_\base (\partial_\alpha^\base v^i)\, d\mu^\alpha_i.
\end{equation}
In the general case, $ \sigma(j^1v) $ is evaluated using a partition of unity.

\subsection{Continuous variational stresses}

Equation \eqref{stressrepmeasures} shows that variational stresses may be as singular as measures. In this section, we restrict our attention to continuous variational stresses, that is, variational stresses for which the measures $\{\mu_i,\,\mu_i^\alpha\}$ are absolutely continuous with respect to some smooth volume form on $\base$.

\begin{defn}
The bundle of \Emph{stress densities} is 
\[
\stds =\Hom(VJ^1Y,(\pi^1)^*\ext^{d}T^*\base).
\]
It is a vector bundle over $J^1Y$, with projection which we denote by $\pi_{\stds}: \stds\to J^1Y$.
For $\xi\in\strain$,
\[
\xi^*\stds = \Hom(\xi^* VJ^1Y,\ext^{d}T^*\base),
\]
is the $ C^0 $-bundle of variational stress densities along $\xi$; it is a vector bundle over $\base$.
\end{defn}

\begin{defn}
A variational stress $\sigma\in T_{\xi}^*\strain$ at $\xi\in\strain$
is termed \Emph{continuous} if there exists a \Emph{variational stress density field} at $\xi$,
\[
\StressDensity\in \Czero(\xi^*\stds)  \simeq \Czero(\Hom(\xi^* VJ^1Y,\ext^{d}T^*\base)),
\]
such that for every velocity jet $\eta$ at $\xi$, the virtual power
that $\sigma$ expends on $\eta$ is given by 
\[
\sigma(\eta)=\int_{\base}\StressDensity\pushf{}\eta.
\]
Note that on the left-hand side, $\eta$ is viewed as an element of
$T_{\xi}\strain$, whereas on the right-hand side, $\eta$ is viewed
as an element of $\Czero(\xi^*VJ^1Y)$. (See below the local expressions for continuous variational stress densities.)
\end{defn}

Let $f$ be a force at $\bk$ and suppose that $f$ is represented by a continuous variational stress at $j^1\bk$, $\sigma\in T^*_{j^1\bk}\strain$ with variational stress density field $\StressDensity\in {C^0(\Hom(VJ^1Y,(\pi^1)^*\ext^{d}T^*\base))}$.
Then, the virtual power expended
by $f$ is given by
\begin{equation}
f(v)= \sigma(j^1v) = \int_{\base}\StressDensity\pushf{}(j^1v).
\label{eq:eq_eq}
\end{equation}

\subsection{Traction stresses}
\label{sec_traction_stress_densities}

In classical formulations of continuum mechanics in a Euclidean space,
the stress object plays two important roles: it determines the traction
fields on sub-bodies via the Cauchy formula, and it acts on velocity jets
to produce power. For continuous stresses on manifolds, two distinct objects play these two roles. 

The variational stress, as defined above
and as its name suggests, produces power when it acts on velocity jets. The object that determines the traction fields on the boundaries of
sub-bodies will be referred to as \Emph{traction stress} (see \cite{Segev2002},
where it is referred to as the Cauchy stress, and \cite{SegevRevMMAS2012},
for the case of a trivial bundle). 

\begin{defn}
The bundle of \Emph{traction stress densities} is
\[
\trstds:=\Hom(VY,\pi^*\ext^{d-1}T^*\base).
\]
It is a vector bundle over $Y$, with projection which we denote by $\pi_{\trstds}:\trstds\to Y$.
For $\bk\in\Q$,
\[
\bk^*\trstds = \Hom(\bk^* VY,\ext^{d-1}T^*\base) ,
\]
is the bundle of traction stresses along $\bk$; it is a vector bundle over $\base$.
\end{defn}

\begin{defn}
A \Emph{traction stress density field} at $\bk$ is a continuous section of the bundle of traction stress densities along $\bk$,
\[
\trstf \in \Czero(\bk^*\trstds) = \Czero(\Hom(\bk^* VY,\ext^{d-1}T^*\base)).
\]
\end{defn}
One would like to restrict traction stress density fields to co-dimension $1$ submanifolds of $\base$. In particular, to the boundary of $\B$.
Consider therefore, an embedded $(d-1)$-dimensional, oriented submanifold $\subm\subset\base$. Denote by $\iota_\subm:\subm\to \base$ the inclusion. We denote by $Y|_\subm$ the restriction of the fiber bundle $Y$ to $\subm$. Formally,
\[
Y|_{\subm}:=(\iota_\subm)^*Y
\Textand
\pi|_\subm:=(\iota_\subm)^*\pi : Y|_\subm \to \subm.
\]
Let $(E,\pi_E,Y)$ be a vector bundle over $Y$ (below $E$ will represent the vector bundles $\trstds$ and $VY$). One can restrict $E$ to $\subm$ by pulling back $E$ over $Y|_\subm$ using the map $\pi^*\iota_\subm$ to obtain the bundle $E|_\subm:=(\pi^*\iota_\subm)^*E$ with the corresponding projection 
\[
\pi_E|_\subm = (\pi^*\iota_\subm)^* \pi_E :E|_\subm\to Y|_\subm;
\]
see the following diagram.
\[
\begin{xy}
(-15,15)*+{\subm} = "1x1";
(25,15)*+{\base} = "2x1";%
(-15,35)*+{Y|_\subm} = "1x2";%
(25,35)*+{Y} = "2x2";%
(-15,55)*+{E|_\subm} = "1x3";
(25,55)*+{E} = "2x3";%
{\ar@{<-}^{\iota_\subm} "2x1"; "1x1"};%
{\ar@{->}^{\pi} "2x2"; "2x1"};%
{\ar@{->}_{\pi|_\subm} "1x2"; "1x1"};%
{\ar@{->}^{\pi_E} "2x3"; "2x2"};%
{\ar@{->}^{\pi^*\iota_{\subm}} "1x2"; "2x2"};%
{\ar@{->}_{\pi_E|_\subm} "1x3"; "1x2"};%
{\ar@{->}^{\pi_E^*\pi^*\iota_{\subm}} "1x3"; "2x3"};%
\end{xy} 
\]

Equipped with this notation scheme,
\begin{equation}
\label{trstdsrestricted}
\begin{split}
\trstds|_\subm &= (\pi^*\iota_\subm)^*\trstds \\
&= \Hom((\pi^*\iota_\subm)^*VY,(\pi^*\iota_\subm)^*\pi^*(\ext^{d-1}T^*\base)) \\
&=\Hom(VY|_\subm,(\pi^*\ext^{d-1}T^*\base)|_\subm).
\end{split}
\end{equation}
The bundle $\trstds|_\subm$ is the bundle of traction stress densities restricted to $\subm$, which nonetheless act on any $(d-1)$-tuple of vectors in $T\base$; it is a vector bundle over $Y|_\subm$.

The inclusion $\iota_\subm:\subm\to\base$ induces a restriction 
\[
(d\iota_\subm)^*: \iota_\subm^*\ext^{d-1}T^*\base\tto \ext^{d-1}T^*\subm
\] 
of $(d-1)$-forms to vectors tangent to $\subm$.  It follows that
\[
\begin{split}
(\pi|_\subm)^*(d\iota_\subm)^* &: (\pi|_\subm)^*\iota_\subm^*\ext^{d-1}T^*\base\tto (\pi|_\subm)^*\ext^{d-1}T^*\subm \\
&: \pi^*\ext^{d-1}T^*\base|_\subm\tto (\pi|_\subm)^*\ext^{d-1}T^*\subm.
\end{split}
\] 
Composition with the latter defines a vector bundle morphism over $Y|_\subm$,
\[
\cres_{\subm}:\trstds\resto{\subm}\tto\Hom(VY|_\subm,(\pi|_\subm)^*\ext^{d-1}T^*\subm).
\]
The mapping $\cres_{\subm}$ is a generalization of the traditional
Cauchy formula for continuum mechanics in Euclidean space, $\tau\mapsto\tau(\mathbf{n})$,
where $\mathbf{n}$ is the unit normal to the oriented submanifold.
We will therefore refer to it as the \Emph{Cauchy mapping}.

Let $\bk\in\Q$. Denote by 
\[
\bk_\subm : \subm \to Y|_\subm
\]
 the restriction of $\bk$ to $\subm$, namely,
\[
\pi^*\iota_\subm\circ \bk_\subm=\bk\circ\iota_\subm.
\]
Note that for every vector bundle $E$ over $Y$,
\[
\bk_\subm^* E|_\subm = \bk_\subm^*(\pi^*\iota_\subm)^*E = \iota_\subm^* \bk^* E = (\bk^*E)|_\subm.
\]

Let $\trstf$ be a traction stress field at $\bk$.
Then, $\iota_\subm^*\trstf$ is a section of 
$(\bk\circ\iota _\subm)^*\trstds$, where
\[
\begin{split}
(\bk\circ\iota _\subm)^*\trstds &\simeq \Hom((\bk\circ\iota _\subm)^*VY,\iota _\subm^*\ext^{d-1}T^*\base) \\
&\simeq \Hom(\bk^*VY|_\subm,\ext^{d-1}T^*\base|_\subm),
\end{split}
\]
and we write $\bk^*VY|_\subm$, rather than $(\bk^*VY)|_\subm$, as no ambiguity should arise.
Moreover, by \eqref{trstdsrestricted}, 
\[
\bk_\subm^*\trstds|_\subm\simeq \Hom(\bk_\subm^*VY|_\subm, \ext^{d-1}T^*\base|_\subm)\simeq (\bk\circ\iota _\subm)^*\trstds .
\]
We may therefore apply the pullback of the Cauchy mapping $\bk_\subm ^*\cres_\subm$ on $\iota_\subm^*\trstf$ to obtain a section 
\[
\SurfaceForceDensity:=(\bk_\subm^*\cres_\subm)(\iota_\subm^*\trstf) \in  \Czero(\Hom(\bk^*VY|_\subm,\ext^{d-1}T^*\subm)).
\]
In particular, for the case  $\subm=\partial\base$, a  traction stress density field $\trstf$ at $\bk$ induces via the Cauchy mapping a surface force density field at $\bk$, that is, a vector bundle morphism
\[
\SurfaceForceDensity= (\bk_{\partial\base}^*\cres_{\partial\base})(\iota_{\partial\base}^*\trstf)\in \Czero(\Hom(\bk^*VY|_{\partial\base},\ext^{d-1}T^*\partial\base)).
\]
In order to simplify the notation, we define the morphism
\[  
\cres_{\partial\base,\bk}:
 \Czero(\Hom(\bk^* VY,\ext^{d-1}T^*\base))
\tto
\Czero(\Hom(\bk^*VY|_{\partial\base},\ext^{d-1}T^*\partial\base))
\]
by
\[ 
\tau\lmt(\bk_{\partial\base}^*\cres_{\partial\base})(\iota_{\partial\base}^*\trstf)
 \]
so that
 \[
 \SurfaceForceDensity= \cres_{\partial\base,\bk}(\trstf).
 \]
Thus, $ \cres_{\partial\base,\bk} $ is the Cauchy mapping along $ \bk $. These formal definitions simply imply that
\begin{equation} 
\label{eq:CauchyRest} 
\SurfaceForceDensity\circ v\resto{{\partial\base}} = 
\cres_{\partial\base,\bk}(\trstf)\circ v\resto{{\partial\base}}
= \iota_{\partial\base}^\sharp (\trstf\resto{{\partial\base}}\circ 
     v\resto{{\partial\base}}).
\end{equation}


\subsection{The traction stress induced by a variational stress density}
\label{tsivs}

The variational stress determines uniquely 
the traction stress, but the converse does not hold. This section
introduces the construction of the traction stress from the variational stress for the fiber bundle
setting. 

The vector bundle $V\pi^{1,0}$ is a sub-bundle of $VJ^1Y$. 
Its elements are represented by paths in $J^1Y$ that are vertical over $Y$ along
the fibers of $\pi^{1,0}$. The fiber bundle $(J^1Y,\pi^{1,0},Y)$
is an affine bundle modeled on the vector bundle \cite[Theorem~4.1.11]{Saunders}
\[
\pi^*T^*\base\otimes_{Y}VY.
\]
Consequently,
\[
V\pi^{1,0}\simeq(\pi^1)^*T^*\base\otimes_{J^1Y}(\pi^{1,0})^*VY\simeq\Hom((\pi^1)^*T\base,(\pi^{1,0})^*VY).
\]
The inclusion $\inclusion:V\pi^{1,0}\inj VJ^1Y$ induces a restriction
\[
\inclusion^*:\stds \tto\Hom(V\pi^{1,0},(\pi^1)^*\LamD),
\]
defined by
\[
\sdv\lmt\sdv\circ\inclusion.
\]
In view of the isomorphism
\[
\Hom(V\pi^{1,0},(\pi^1)^*\LamD)\simeq (\pi^{1,0})^*\brk{
\pi^*T\base\otimes_{Y}(VY)^*\otimes_{Y}\pi^*\LamD},
\]
we define
\[
\contr:\Hom(V\pi^{1,0},(\pi^1)^*\LamD)\tto (\pi^{1,0})^*\brk{(VY)^*\otimes_{Y}\pi^*\ext^{d-1}T^*\base} \simeq (\pi^{1,0})^*\trstds
\]
to be the contraction of the third and first factors in the product, namely,
\[
\contr(v\otimes\varphi\otimes\theta)=\varphi\otimes(v\inc\theta).
\]
Note that $\contr$ is a vector bundle morphism over $J^1Y$.

Next, define 
\[
P:\stds \tto(\pi^{1,0})^*\trstds,
\]
a vector bundle morphism over $J^1Y$, by
\begin{equation}
P=\contr\circ\inclusion^*.
\end{equation}
For $\bk\in\Q$, 
\[
P_\bk := (j^1\bk)^*P: (j^1\bk)^*\stds \tto \bk^*\trstds
\]
maps variational stress densities at $j^1\bk$ to traction stress densities at $\bk$.
It follows that composition with $P_\bk$ maps variational stress density fields along $j^1\bk$ to traction stress fields along $\bk$.

\InCoord{ An element of $(\jk)^*\stds$ at $p$
is of the form 
\[
\StressDensity_ p =\brk{a_i\dxJY^i|_{j_p^1\bk}+a_i^{\alpha}(\dxJY)_{\alpha}^i|_{j_p^1\bk}}\otimes\Vol_p.
\]
Its restriction to $(V\pi^{1,0})_{j_p^1\bk}$, it can be written
as 
\[
a_i^{\alpha}\,\parB_{\alpha}|_p\otimes\dxY^i|_{\bk_p}\otimes\Vol_p.
\]
Then, 
\[
P_\bk (\StressDensity_p) = a_i^{\alpha}\,\dxY^i|_{\bk_p}\otimes(\parB_{\alpha}\inc\Vol)_p
\]
Thus, for a variational stress density at $\jk$ given locally by
\[
\StressDensity=\brk{\StressDensity_i\,(\jk)^*\dxJY^i+\StressDensity_i^{\alpha}\,(\jk)^*(\dxJY)_{\alpha}^i}\otimes\Vol,
\]
we have 
\[
P_\bk \pushf{}\StressDensity=\StressDensity_i^{\alpha}\,\bk^*\dxY^i\otimes(\parB_{\alpha}\inc\Vol).
\]
} 

\subsection{The exterior jet of a differentiable traction stress density}

%
%
Recall the definition of a linear differential operator (see \cite[Chapter 3]{Palais68}).

\begin{defn}
Let $(E,\pi_E,\man)$ and $(F,\pi_F,\man)$ be vector bundles over the same base manifold $\man$. A linear map $D:C^k (E)\to C^{k-1}(F)$ is called a \Emph{first-order linear differential operator} from $E$ to $F$ if there exists a vector bundle morphism $\tilde{D}\in L(J^1E,F)$, such that
\[
D(s)= \tilde{D}\circ j^1s,\qquad \forall s\in C^k(E).
\]
Note that $J^1E$ is a vector bundle over $\man$ only if $E$ is a vector bundle. We will refer to $ \tilde{D} $ as the morphism associated with $ D $.
\end{defn}

For convenience, we recall the definitions of the following vector bundles: 
\[
\begin{aligned}
& \text{body force densities} &\qquad  \vals B &= \Hom(VY,\pi^*\LamD), \\
& \text{surface force densities} &\qquad \vals T &= \Hom(VY|_{\partial\base},(\pi|_{\partial\base})^*\ext^{d-1}T^*\partial\base), \\
& \text{variational stress densities} &\qquad \stds &= \Hom(VJ^1Y,(\pi^1)^*\LamD), \\
& \text{traction stress densities} &\qquad \trstds &= \Hom(VY,\pi^*\ext^{d-1}T^*\base).
\end{aligned}
\]
Each of these bundles can be pulled back with either $\bk$ or $j^1\bk$; sections of the pullback bundles are referred to as fields along $\bk$. 

\begin{defn}
Let $\bk\in\Q$.
The \Emph{exterior jet} differential operator along $\bk$ 
\[
\frakd_\bk:\Cone(\bk^*\trstds) \tto  \Czero((j^1\bk)^*\stds),
\]
is defined as follows. Let $\trstf \in \Cone(\bk^*\trstds)$ be a traction
stress density field along $\bk$. Then, for all $v\in \Cone(\bk^*VY)$,
\[
(\frakd_\bk \trstf)\circ j^1v=d(\trstf\pushf v).
\]
\end{defn}

\begin{prop}
$\frakd_\bk$ is a well-defined first-order linear differential operator. 
\end{prop}
\begin{proof}
In local coordinates, a  traction
stress density field along $\bk$ takes the form
\[
\trstf = \trstf_i^{\alpha}\,\bk^*\dxY^i\otimes(\parB_{\alpha}\inc\Vol),
\]
where
\[
dX = \dxB^1\wedge\cdots\wedge \dxB^d.
\]
A velocity at $\bk$ is given locally by 
\[
v=v^i\,\bk^*\parY_i,
\]
hence, 
\[
\trstf \pushf v = \trstf_i^{\alpha}v^i\,(\parB_{\alpha}\inc\Vol).
\]
Taking the exterior derivative 
\[
d(\trstf\pushf v)=\brk{(\parB_{\alpha}\trstf_i^{\alpha})v^i+\trstf_i^{\alpha}\,(\parB_{\alpha}v^i)}\,\Vol.
\]
By the local expression \eqref{eq:j1v} for $j^1v$, 
it follows that
\[
\frakd_\bk \trstf=\brk{(\parB_{\alpha}\trstf_i^{\alpha})\,(\jk)^*\dxJY^i+\trstf_i^{\alpha}\,(\jk)^*(\dxJY)_{\alpha}^i}\otimes\Vol
\]
satisfies the defining property of $\frakd_\bk \trstf$ and depends linearly on $\trstf$ and its derivatives.
\end{proof}

\subsection{The divergence of differentiable variational stress densities and the equilibrium field equations}
\label{defe}

In this section we define the divergence differential operator for differentiable variational stress densities. The divergence operator enables one to tranform the weak form of the compatibility condition between continuous forces and stresses to a strong form of the equilibrium differential equation. For the rest of this section, we restrict ourselves to $C^2$-configurations and $C^1$-stress density fields.

\begin{prop}
\label{prop:divergence}
Let $\bk$ be a $C^2$-configuration. There exists a first-order linear differential operator mapping differentiable variational stress density fields along $j^1\bk$ into continuous body force density fields along $\bk$,
\[
\divergence_\bk: \Cone((\jk)^* \stds) \tto \Czero(\bk^*\vals B),
\]
satisfying for every variational stress density field $\StressDensity\in\Cone((\jk)^* \stds)$ and virtual velocity $v\in C^1(\bk^*VY)$,
\begin{equation}
(\divergence_\bk\StressDensity)\circ v=\frakd_\bk(P_\bk \pushf\StressDensity)\pushf j^1v -\StressDensity\pushf j^1v.
\label{eq:def_divergence}
\end{equation}

\end{prop}

\begin{proof}
Using coordinates, a variational stress density field has a local representation in the form
\[
\StressDensity=\brk{\StressDensity_i\,(\jk)^*\dxJY^i+\StressDensity_i^{\alpha}\,(\jk)^*(\dxJY)_{\alpha}^i}\otimes\Vol,
\]
and a velocity has a local representation 
\[
v=v^i\,\bk^*\parY_i.
\]
By \eqref{eq:j1v}, 
\[
\StressDensity\pushf j^1v=\brk{\StressDensity_iv^i+\StressDensity_i^{\alpha}\,(\parB_{\alpha}v^i)}\,\Vol.
\]
On the other hand, 
\[
P_\bk\pushf\StressDensity=\StressDensity_i^{\alpha}\,\bk^*\dxY^i\otimes(\parB_{\alpha}\inc\Vol),
\]
and 
\[
\frakd_\bk(P_\bk\pushf\StressDensity)\circ j^1v=\brk{(\parB_{\alpha}\StressDensity_i^{\alpha})v^i+\StressDensity_i^{\alpha}\,(\parB_{\alpha}v^i)}\,\Vol.
\]
Subtracting, we obtain that 
\[
(\divergence_\bk\StressDensity)\pushf v=\brk{(\parB_{\alpha}\StressDensity_i^{\alpha})v^i-\StressDensity_iv^i}\,\Vol.
\]
That is, 
\begin{equation}
\divergence_\bk\StressDensity=\brk{(\parB_{\alpha}\StressDensity_i^{\alpha})-\StressDensity_i}\,\bk^*\dxY^i\otimes\Vol.
\label{eq:div_rep}
\end{equation}
\end{proof}

Consider now a continuous force $f\in T_\bk^*\Q$, given by
\[
f(v)=\int_{\base}\ForceDensity\pushf{}v+\int_{\partial\base}\SurfaceForceDensity\pushf{}v|_{\partial\base},
\]
where $\ForceDensity\in \Czero(\bk^*\vals B)$ and $\SurfaceForceDensity\in \Czero(\bk_{\partial\base}^*\vals T)$. Consider further a continuous stress $\sigma$ represented by a variational stress density field $\StressDensity\in \Cone((j^1\bk)^*\stds)$. Then, \eqref{eq:eq_eq} takes the form
\[
\begin{split}f(v) & =
\sigma(j^1v) \\ 
&=\int_{\base}\StressDensity\pushf{}(j^1v)\\
 & =-\int_{\base}(\divergence_\bk\StressDensity)\pushf v+\int_{\base}(\frakd_\bk (P_\bk\pushf\StressDensity))\pushf{}(j^1v)\\
 & =-\int_{\base}(\divergence_\bk\StressDensity)\pushf v+\int_{\base}d((P_\bk\pushf{}\StressDensity)\pushf{}v)\\
 & =-\int_{\base}(\divergence_\bk\StressDensity)\pushf v+\int_{\partial\base}
 (P_\bk\pushf{}\StressDensity)\pushf{}v.
\end{split}
\]

 
We conclude that $f$ is represented by a variational stress density field $\StressDensity$ if for every $v\in C^1(\bk^*VY)$
\[
\int_{\base}\ForceDensity\pushf{}v+\int_{\partial\base}\SurfaceForceDensity\pushf{}v|_{\partial\base}=-\int_{\base}(\divergence_\bk\StressDensity)\pushf v+\int_{\partial\base}
(P_\bk\pushf{}\StressDensity)\pushf{}v.
\]
It is observed that while the restriction of the various terms in the integrand $ (P_\bk\pushf{}\StressDensity)\pushf{}v $ to $ \partial\base $ and to vectors tangent to $ \partial\base $ is implied in integration theory, formally, in view of (\ref{eq:CauchyRest}) we write the integrand as 
$\cres_{\partial\base,\bk}(P_\bk\circ\StressDensity)\circ v$.
Thus, since equality holds for every $v$, we obtain  
\begin{equation}
\label{eq1000}
\divergence_\bk\StressDensity+\ForceDensity=0\,\, \text{on }\base\quad\text{and }\quad  \SurfaceForceDensity=
\cres_{\partial\base,\bk}(P_\bk\pushf{}\StressDensity)\,\,\text{on }\partial\base.
\end{equation}
In coordinates, equation \eqref{eq1000} transforms to an underdetermined set of $d$ equations for the $(d m+m)$ components of $\StressDensity$.

Note that the above proposition holds  if the
vector bundle $\bk^*VY\to\base$ is replaced by an arbitrary vector bundle $\rho:W\to\base$, in which
case $\divergence:L(J^1\rho,\LamD)\tto L(W,\LamD)$, which is compatible
with \cite{Segev2002}.


\section{The Continuum Mechanics Problem}
\label{sec:Constitutive}

It is well known that a given force distribution, does not determine
a unique stress distribution. The source of this non-uniqueness may
be traced back to Section~\ref{sec5p1}.
A unique stress field is determined when additional
information in the form of a constitutive relation is provided. This
couples the statics problem described above with the kinematics. 
The
basic notions corresponding to the introduction of constitutive relations
are discussed in this section.

Generally, the continuum mechanics problem takes the following form: the system under consideration is subject to external forces, usually dictated by a loading, which is an assignment of a force to every admissible configuration. A loading may consist of a body force component and a surface force component, and it may be singular or continuous.
In analogy, a constitutive relation assigns stresses, regular or singular, to deformation jets, in particular, to jets of configurations. The continuum mechanics problem seeks a configuration $ \bk $ such that the stress associated through the constitutive relation with $ j^1\bk $ represents the force assigned by the loading to $ \bk $. 

\subsection{Loadings}

\begin{defn}
\label{def:loading} A \Emph{loading}, $F$, is a one-form on the configuration
space assigning a force $F_{\bk}\in T_{\bk}^*\Q$ to every configuration $\bk\in\Q$.
\end{defn}


Definition~\ref{def:smooth_force} of continuous forces extends to a
definition of continuous loadings:

\begin{defn}
\label{def:smooth_loading} A loading $F$ is termed \Emph{continuous}
if there exists a \Emph{body loading density}
\[
\LoadingDensity\in \Czero(\pi_{\vals B}),
\]
and a \Emph{surface
loading density} 
\[
\SurfaceLoadingDensity\in \Czero(\pi_{\vals T}),
\]
such that for every $\bk\in\Q$ and $v\in T_{\bk}\Q\simeq\Cone(\bk^*VY)$,
\[
F_{\bk}(v)=\int_{\base} \bk^*\LoadingDensity\pushf v+
\int_{\partial\base}\bk_{\partial\base}^*\SurfaceLoadingDensity\pushf v|_{\partial\base},
\]
That is, $F_{\bk}$ is continuous with body force density field $\bk^*\LoadingDensity\in \Czero(\bk^*\vals B)$
and surface force density field $\bk_{\partial\base}^*\SurfaceLoadingDensity\in \Czero(\bk_{\partial\base}^*\vals T)$.
\end{defn}

Note that by our conventions, we write $\Czero(\pi_{\vals B})$ rather than $\Czero(\vals B)$, because $\vals B$ has multiple bundle structures, and the domain of those sections is not the base manifold, $\base$.

\InCoord{A body loading density has local representation, 
\[
\LoadingDensity=\LoadingDensity_i\,\dxY^i\otimes_{Y}\pi^*\Vol
\]
where $\LoadingDensity_i$ are real valued continuous functions defined locally on $Y$. 

Assume that the chart $X$ is adapted to $\partial\base$ so that
$(X^1,\dots,X^{d-1})$ is a chart on $\partial\base$ and $\parB_{d}$
is transversal to $T\partial\base$. Then, one has a volume element $\parB_{d}\inc\Vol$
on $\partial\base$, and locally,
\[
\SurfaceLoadingDensity=\SurfaceLoadingDensity_i\,\dxY^i\otimes_{Y\resto{\partial\base}}\pi\resto{\partial\base}^*(\parB_{d}\inc\Vol)
\]
where $\SurfaceLoadingDensity_i$ are continuous functions defined locally on $Y|_{\partial\base}$.
Let 
\[
v=v^i\,\bk^*\parY_i
\]
be a local representation of a velocity at $\bk$. As 
\[
\bk^*\LoadingDensity=(\bk^*\LoadingDensity_i)\,\bk^*\dxY^i\otimes_{\base}\Vol
\textand
\bk_{\partial\base}^*\SurfaceLoadingDensity=\bk_{\partial\base}^*\SurfaceLoadingDensity_i\,\bk_{\partial\base}^*\dxY^i\otimes_{\partial\base}(\parB_{d}\inc\Vol).
\]
we have the local expression, 
\[
(\bk^*\LoadingDensity)\pushf v=v^i\,(\bk^*\LoadingDensity_i)\,\Vol,\qquad\bk_{\partial\base}^*\SurfaceLoadingDensity\circ(v_{\partial\base})=v^i\,\bk_{\partial\base}^*\SurfaceLoadingDensity_i(\parB_{d}\inc\Vol).
\]
If the support of $v$ may be covered by a single chart, we may write $F_\bk$ explicitly,
\[
F_{\bk}(v)=\int_{\base}v^i\,(\bk^*\LoadingDensity_i)\,\Vol+
\int_{\partial\base}v^i\,(\bk_{\partial\base}^*\SurfaceLoadingDensity_i)(\parB_{d}\inc\Vol).
\]
} 

\begin{defn}
\label{def:conservative_loading} 
A loading $F$ 
is  termed \Emph{conservative} if
\[
F=-dW,
\]
 for some \Emph{loading potential},
\[
W\in \Cone(\Q).
\]
\end{defn}

\begin{defn}
A conservative loading is \Emph{continuous} if there exists a \Emph{body
loading potential density}, which is a $\Cone$-section
\begin{equation}
w_{\base}: Y\to \pi^*\LamD,
\label{eq:wbase}
\end{equation}
and a \Emph{boundary loading potential density}, which is a $\Cone$-section
\begin{equation}
w_{\partial\base}: Y\resto{\partial\base} \to (\pi|_{\partial\base})^*\ext^{d-1}T^*\partial\base,
\label{eq:wparbase}
\end{equation}
such that the loading potential is given by 
\[
W(\bk)=\int_{\base}\bk^*w_{\base}+\int_{\partial\base}\bk_{\partial\base}^*w_{\partial\base}
\]
(see diagrams).

\[ \begin{xy}
(0,0)*+{\base}="X1";
(30,0)*+{Y}="Y";
(0,20)*+{\LamD}="TX1";
(30,20)*+{\pi^*\LamD}="TY";
{\ar@{-->}@/_{1pc}/_{\bk}"X1";"Y"};
{\ar@{->}_{\pi}"Y";"X1"};
{\ar@{->}^{(\tau^*_{\base})^*{\bk}}"TX1";"TY"};
{\ar@{->}^{\tau^*_{\base}}"TX1";"X1"};
{\ar@{->}_{\pi^*\tau^*_{\base}}"TY";"Y"};
{\ar@{-->}@/_{1pc}/_{{w_\base}}"Y";"TY"};
{\ar@{-->}@/^{1pc}/^{{\bk^*w_{\base}}}"X1";"TX1"};
\end{xy} 
\quad
\begin{xy}
(0,0)*+{\partial\base}="X1";
(40,0)*+{\iota_{\partial\base}^*Y}="Y";
(0,20)*+{\hspace*{-4mm}%
            \ext^{d-1}T^*\partial\base}="TX1";
(40,20)*+{(\iota_{\partial\base}^*\pi)^*
        \ext^{d-1}T^*\partial\base}="TY";
{\ar@{-->}@/_{1pc}/_{\bk_{\partial\base}}"X1";"Y"};
{\ar@{->}_{\iota_{\partial\base}^*\pi}"Y";"X1"};
{\ar@{->}^(0.37){(\tau^*_{\base})^*{\bk}}"TX1";"TY"};
{\ar@{->}^{\tau^*_{\partial\base}}
     "TX1";"X1"};
{\ar@{->}_{(\iota_{\partial\base}^*\pi)^*
        \tau^*_{\partial\base}}"TY";"Y"};
{\ar@{-->}@/_{1pc}/_(0.5){{w_{\partial\base}}}"Y";"TY"};
{\ar@{-->}@/^{1pc}/^{{\bk_{\partial\base}^*w_{\base}}}"X1";"TX1"};
\end{xy}
\] 
\end{defn}

\begin{prop}
Consider a continuous conservative loading $F$, the potential function of which is given by a body loading potential density $w_{\base}$ and boundary loading potential $w_{\partial\base}$.
Then, $F$ is a continuous loading with corresponding loading densities
\[
\LoadingDensity=-dw_{\base}|_{VY},\qquad\SurfaceLoadingDensity=-dw_{\partial\base}|_{VY|_{\partial\base}}.
\]
\end{prop}


\begin{proof}
First, we should verify that this relation is well-defined type-wise.
%
%
By the types \eqref{eq:wbase} and \eqref{eq:wparbase} of $w_\base$ and $w_{\partial\base}$, 
it follows that 
\[
\begin{split}
dw_{\base} & \in \Czero(\Hom(TY,w_{\base}^*T(\pi^*\LamD))),\\
dw_{\partial\base} & \in \Czero(\Hom(TY|_{\partial\base},w_{\partial\base}^*T((\pi|_{\partial\base})^*\ext^{d-1}T^*\partial\base))).
\end{split}
\]
Restricting to the vertical bundle, 
\[
\begin{split}dw_{\base}|_{VY} & \in \Czero(\Hom(VY,w_\base^*V(\pi^*\ext^{d}\tau_{\base}^*)))\\
dw_{\partial\base}|_{VY|_{\partial\base}} & \in \Czero(\Hom(VY|_{\partial\base},w_{\partial\base}^*V((\pi|_{\partial\base})^*\ext^{d-1}\tau^*_{\partial\base}))).
\end{split}
\]
For every vector bundle $(E,\pi_E,Y)$, section $w:Y\to E$
and $a\in Y$, 
\[
(w^*V\pi_E)_{a}\simeq(V\pi_E)_{w(a)}\simeq E_{a}.
\]
In other words, $w^*V\pi_E\simeq E$, and so, 
\[
\begin{split}
\Hom(VY,w_\base^*V(\pi^*\ext^{d}\tau_{\base}^*)) & \simeq \Hom(VY,\pi^*\LamD),\\
\Hom(VY|_{\partial\base},w_{\partial\base}^*V((\pi|_{\partial\base})^*\ext^{d-1}\tau^*_{\partial\base})) & \simeq \Hom(VY|_{\partial\base},(\pi|_{\partial\base})^*\ext^{d-1}T^*\partial\base),
\end{split}
\]
and so, $dw|_{VY}$ and $dw_{\partial\base}|_{VY|_{\partial\base}}$
are indeed loading densities.

It remains to show that $dw_{\base}|_{VY}$ and $dw_{\partial\base}|_{VY|_{\partial\base}}$
are the loading densities corresponding to $F=-dW$. By definition,
for $v\in T_{\bk}\Q\simeq\Cone(\bk^*VY)$, 
\[
\begin{split}dW_{\bk}(v)=\left.\deriv{}t\right|_{0}W\circ\gamma,\end{split}
\]
where $\gamma:I\to\Q$ satisfies $\gamma(0)=\bk$ and $\dot{\gamma}(0)=v$.
Substituting the representation of $W$ in terms of a density, 
\[
\begin{split}dW_{\bk}(v) & =\left.\deriv{}t\right|_{0}\left[\int_{\base}(\gamma(t))^*w_{\base}+\int_{\partial\base}(\gamma(t)_{\partial\base})^*w_{\partial\base}\right],\\
 & =\int_{\base}\left.\deriv{}t\right|_{0}(\gamma(t))^*w_\base+\int_{\partial\base}\left.\deriv{}t\right|_{0}(\gamma(t)_{\partial\base})^*w_{\partial\base},
\end{split}
\]
where the derivatives $\left.\deriv{}t\right|_{0}(\gamma(t))^*w_{\base}$
and $\left.\deriv{}t\right|_{0}(\gamma(t)_{\partial\base})^*w_{\partial\base}$
are taken pointwise in $\base$ and $\partial\base$,
respectively.

It follows from the chain rule and the definition of $\gamma$ that 
\[
\left.\deriv{}t\right|_{0}((\gamma(t))^*w_{\base})_{p}=(dw_{\base})_{\bk_{p}}\circ(\dot{\gamma}(0))_{p}=(\bk^*dw_{\base}|_{VY}\circ v)_{p},
\]
and
\[
\left.\deriv{}t\right|_{0}((\gamma(t)_{\partial\base})^*w_{\partial\base})_{q}=(dw_{\partial\base})_{\bk_{q}}\circ(\dot{\gamma}(0))_{q}=((\bk_{\partial\base})^*dw_{\partial\base}|_{VY|_{\partial\base}}\circ v)_{q}.
\]
Thus, 
\[
dW_{\bk}(v)=\int_{\base}(\bk^*dw_{\base}|_{VY})\pushf v+\int_{\partial\base}((\bk_{\partial\base})^*dw_{\partial\base}|_{VY|_{\partial\base}})\pushf v,
\]
which completes the proof. 
\end{proof}


\InCoord{ In a local coordinate frame, a loading potential density takes
 the form, 
\[
w_{\base}=\omega_{\base}\,\pi^*\Vol
\Textand 
w_{\partial\base}=\omega_{\partial\base}\,\pi\resto{\partial\base}^*(\parB_{d}\inc\Vol).
\]
where $\omega_{\base}$ and $\omega_{\partial\base}$ are differentiable
functions defined on the domains of the adapted charts on $\base$
and $\partial\base$, respectively. Since the exterior derivatives
of the volume elements vanish identically, the loading densities are
given locally by
\[
\LoadingDensity=-\partial_i\omega_{\base}\,\,\dxY^i\otimes_{Y}\pi^*\Vol,\qquad\SurfaceLoadingDensity=-\partial_i\omega_{\partial\base}\,\dxY^i\otimes_{Y\resto{\partial\base}}\pi\resto{\partial\base}^*(\parB_{d}\inc\Vol).
\]
} 


\subsection{Constitutive relations}

All the notions analyzed above\textemdash configurations, velocities,
forces, {\strn}s and stresses\textemdash pertain to a particular
theory depending on the choice of configuration space. Within a particular
theory, the  properties of a system are prescribed using a constitutive
relation, associating the state-of-stress with the kinematics.

It is noted that since our base space $\base$ may be interpreted
naturally as physical space-time, constitutive relations may involve time
even in cases where it is not explicitly indicated.

\begin{defn}
A \Emph{constitutive relation} is a one-form $\crel:\strain\to T^*\strain$ on the manifold
of deformation jets, 
assigning a stress $\crel_{\xi}\in T_{\xi}^*\strain$ to every deformation jet $\xi\in\strain$. 
\end{defn}

Since $j^1:\Q\to\strain$ is injective, a constitutive relation
$\crel$ induces a loading 
\[
(j^1)^{\push}\crel:\Q\tto T^*\Q.
\]
By the definition of the pullback of forms, for $v\in T_{\bk}\Q$,
\[
((j^1)^{\push}\crel)_{\bk}(v)=((j^1)^*\crel)_{\bk}(j^1v)=\crel_{\jk}(j^1v),
\]
where in the first identity we identified, as before, $dj^1(v)$ with $j^1v$.

We note that at this level of generality, and interpreting $\base$
as space-time, this, seemingly elastic constitutive relation, is actually
non-local, time dependent, and not necessarily causal.

The definition of continuous stresses extends readily to a definition
of $C^k$-differentiable constitutive relations:

\begin{defn}
A constitutive relation $\crel:\strain\to T^*\strain$ is termed $C^k$-\Emph{differentiable}, $ k=0,1,\dots $,
if there exists a \Emph{constitutive density}, a differentiable
section, 
\[
\ConstitutiveDensity\in C^k(\pi_{\stds}),
\]
such that for every {\strn} $\xi\in\strain$ and for every {\vjet}
 $\eta$ at $\xi$, 
\[
\crel_{\xi}(\eta)=\int_{\base}(\xi^*\ConstitutiveDensity)\pushf{}\eta.
\]
\end{defn}

In the case of a differentiable constitutive relation, the induced loading is given by
\[
((j^1)^\push\crel)_\bk(v)=\crel_{\jk}(j^1v)=\int_{\base}((\jk)^*\ConstitutiveDensity)\pushf{}(j^1v),\quad v\in T_{\bk}\Q.
\]
We note that a differentiable constitutive relation is local and elastic in
the sense that the stress at a point depends only on the deformation jet
at that point. Again, we mention that for the case where $\base$
is interpreted as space-time or body-time, $\xi$ contains components
that may reflect velocities. However, history dependence is excluded.

\InCoord{ In coordinates, a constitutive density is
represented locally as 
\[
\ConstitutiveDensity=\brk{\ConstitutiveDensity_i\,\dxJY^i+\ConstitutiveDensity_i^{\alpha}\,(\dxJY)_{\alpha}^i}\otimes_{J^1Y}(\pi^1)^*\Vol,
\]
where $\ConstitutiveDensity_i,\ConstitutiveDensity_i^{\alpha}$ are $C^k$ over the domain of a chart in $J^1Y$.
} 

\subsection{The continuum mechanics equations}

Given a constitutive relation $\crel$ and a loading $F$, the continuum
mechanics problem seeks a configuration for which the stress determined
by the constitutive relation represents the loading for that configuration.
Thus, the condition is expressed by
\begin{equation}
(j^1)^{\push}\crel=F.
\label{eq:eq_eq-1}
\end{equation}
Explicitly, a configuration $\bk\in\Q$ is a solution of \eqref{eq:eq_eq-1}
if for every virtual velocity $v$ at $\bk$, 
\[
\crel_{\jk}(j^1v)=F_{\bk}(v).
\]

\subsection{The differential form of the continuum mechanics equations}

Consider the continuum mechanics problem in the case of a $C^k$-constitutive density $\ConstitutiveDensity$
and a continuous loading with densities $\LoadingDensity$ and
$\SurfaceLoadingDensity$. A configuration $\bk\in\Q$ which solves
the continuum mechanics problem satisfies
\begin{equation}
\int_{\base} (j^1\bk)^*\ConstitutiveDensity\circ j^1v =
\int_{\base}\bk^* \LoadingDensity\pushf v+
\int_{\partial\base}\bk_{\partial\base}^*\SurfaceLoadingDensity\pushf(v\resto{\partial\base}).
\label{eq:integrand}
\end{equation}
for every velocity field $v\in T_{\bk}\Q\simeq\Cone(\bk^*VY)$.

In view of the definition of the divergence, restricting ourselves to $C^2$-configurations and $C^1$ constitutive densities, we can rewrite these
equations as follows: 
\begin{multline*}
-\int_{\base}\divergence_\bk ((j^1\bk)^*\ConstitutiveDensity)\pushf v 
+ \int_{\partial\base}
((j^1\bk)^*(P\pushf \ConstitutiveDensity))\pushf v 
\\
=\int_{\base}\bk^*\LoadingDensity\pushf v + \int_{\partial\base}\bk_{\partial\base}^* \SurfaceLoadingDensity\pushf(v\resto{\partial\base}).
\end{multline*}
Since this hold for every virtual velocity $v$, we obtain a boundary-value
problem 
\[
\begin{gathered}
\divergence_\bk ((j^1\bk)^*\ConstitutiveDensity) + \bk^*\LoadingDensity=0,\qquad\text{in \ensuremath{\base},}\\
\cres_{\partial\base,\bk}
\circ ((j^1\bk)^*(P\pushf \ConstitutiveDensity))
= 
\bk_{\partial\base}^*\SurfaceLoadingDensity,\qquad\text{on \ensuremath{\partial\base}.}
\end{gathered}
\]

\InCoord{ Let the loading densities and the constitutive density
be given locally by 
\[
\LoadingDensity=\LoadingDensity_i\,\dxY^i\otimes_{Y}\pi^*\Vol,\qquad\SurfaceLoadingDensity=\SurfaceLoadingDensity_i\,\dxY^i\otimes_{Y\resto{\partial\base}}\pi\resto{\partial\base}^*(\parB_{d}\inc\Vol),
\]
and 
\[
\ConstitutiveDensity=\brk{\ConstitutiveDensity_i\,\dxJY^i+\ConstitutiveDensity_i^{\alpha}\,(\dxJY)_{\alpha}^i}\otimes_{J^1Y}(\pi^1)^*\Vol,
\]
Then, 
\[
\bk^*\LoadingDensity=(\bk^*\LoadingDensity_i)\,\bk^*\dxY^i\otimes_{X}\Vol,
\qquad
\bk_{\partial\base}^*\SurfaceLoadingDensity=(\bk\resto{\partial\base})^*\SurfaceLoadingDensity_i\,(\bk\resto{\partial\base})^*\dxY^i\otimes_{\partial\base}(\parB_{d}\inc\Vol)
\]
and 
\[
(j^1\bk)^*\ConstitutiveDensity = \brk{((\jk)^*\ConstitutiveDensity_i)\,(\jk)^*\dxJY^i+((\jk)^*\ConstitutiveDensity_i^{\alpha})\,(\jk)^*(\dxJY)_{\alpha}^i}\otimes_{J^1Y}\Vol.
\]
If follows from the local expressions derived in Sections \ref{tsivs} and \ref{defe} that 
\[
(j^1\bk)^*(P\pushf\ConstitutiveDensity) =
((j^1\bk)^*\ConstitutiveDensity_i^\alpha)\,\bk^*dx^i_Y\otimes (\partial^\base_\alpha\inc\Vol),
\]
and
\[
\divergence_\bk(j^1\bk\ConstitutiveDensity)=\brk{(\jk)^*\ConstitutiveDensity_i-\parB_{\alpha}((\jk)^*\ConstitutiveDensity_i^{\alpha})}\,\bk^*\dxY^i\otimes\Vol,
\]
so that the field equation is 
\[
(\jk)^*\ConstitutiveDensity_i-\parB_{\alpha}((\jk)^*\ConstitutiveDensity_i^{\alpha})+(\bk^*\LoadingDensity_i)=0.
\]
Since
\begin{equation}
\iota_{\partial\base}^\push(\parB_{\alpha}\inc\Vol)=\begin{cases}
0, & \text{for }\alpha\ne d,\\
\parB_{d}\inc\Vol, & \text{for }\alpha=d,
\end{cases}
\end{equation}
the boundary conditions are
\[
(\jk)^*\ConstitutiveDensity_i^{d}=(\bk\resto{\partial\base})^*\SurfaceLoadingDensity_i.
\]
} 

\subsection{Conservative constitutive relations}

\begin{defn}
A constitutive relation is referred to as \Emph{conservative} if there exists an energy
function 
\[
\sten\in \Cone(\strain)
\]
such that the constitutive relation is given by 
\[
\crel=d\sten.
\]
\end{defn}

Consider a conservative system with energy function $\sten$, subject
to a conservative loading $F=-dW$. Then, the equilibrium equations
are 
\[
\begin{split}
(j^1)^{\push}\crel-F & =(j^1)^*\crel\circ dj^1-F\\
 & =(j^1)^*d\sten\circ dj^1+dW\\
 & =d(\sten\circ j^1+W)=0,
\end{split}
\]
where we used the commutativity of external derivatives and pullbacks. Thus, the solution
is a critical point of the \Emph{total energy}, 
\begin{equation}
\sten\circ j^1+W.\label{eq:calF}
\end{equation}

\begin{defn}
A conservative constitutive relation is said to be \Emph{differentiable} if there exists a \Emph{Lagrangian
density form},
a $C^2$-section $$\calL:J^1Y\to (\pi^1)^*\LamD,$$
such that the energy function is of the form 
\[
\sten(\xi)=\int_{\base}\xi^*\calL.
\]
\end{defn}
\[
\begin{xy}(0,0)*+{\base}="X";
(30,0)*+{J^1Y}="J1Y";
(0,20)*+{\LamD}="Vol";
(30,20)*+{(\pi^1)^*\LamD}="J1Vol";
{\ar@{->}_{\xi}"X";"J1Y"};
{\ar@{->}_{\tau_{\base}^*}"Vol";"X"};
{\ar@{->}^{}"Vol";"J1Vol"};
{\ar@{->}^{(\pi^1)^*\tau_{\base}^*}"J1Vol";"J1Y"};
{\ar@{-->}@/_{4pc}/_{{\calL}}"J1Y";"J1Vol"};
{\ar@{-->}@/^{3pc}/^{{\xi^*\calL}}"X";"Vol"};
\end{xy}
\]

Evidently, differentiable conservative constitutive relations generalize hyperelastic constitutive relations to the setting considered in this paper.

\begin{prop}
Consider a differentiable constitutive relation with a Lagrangian density $\calL$.
The resulting constitutive relation is $C^1$ with a constitutive density
\begin{equation}
\ConstitutiveDensity=d\calL|_{VJ^1Y}.\label{eq:smooth_hyper}
\end{equation}
Equation \eqref{eq:smooth_hyper} is a generalization of the classical relation between
the stress and the derivative of the elastic energy density. 
\end{prop}
\begin{proof}
By definition, for $\eta\in T_{\jk}\strain\simeq\Czero((\jk)^*VJ^1Y)$,
\[
\begin{split}d\sten_{\jk}(\eta)=\left.\deriv{}t\right|_{0}\sten(\gamma(t)),\end{split}
\]
where $\gamma:I\to\strain$ satisfies $\gamma(0)=\jk$ and $\dot{\gamma}(0)=\eta$.
Substituting the representation of $\sten$ in terms of a density,
\[
d\sten_{\jk}(\eta)=\left.\deriv{}t\right|_{0}\int_{\base}(\gamma(t))^*\calL=\int_{\base}\left.\deriv{}t\right|_{0}(\gamma(t))^*\calL,
\]
where the derivative $\left.\deriv{}t\right|_{0}(\gamma(t))^*\calL$
is taken pointwise at every $p\in\base$. It follows from the chain rule
and the definitions of $\gamma$ and $\eta$ that 
\[
\left.\deriv{}t\right|_{0}(\gamma(t))^*\calL)_{p}=d\calL_{j_{p}^1\bk}\circ(\dot{\gamma}(0))_{p}=(d\calL|_{VJ^1Y})_{j_{p}^1\bk}\circ\eta_{p},
\]
\ie, 
\[
d\sten_{\jk}(\eta)=\int_{\base}((\jk)^*d\calL|_{VJ^1Y})\circ\eta,
\]
which completes the proof. 
\end{proof}
\InCoord{ Using a coordinate frame, a Lagrangian density is of the form
\[
\calL=L\,\,(\pi^1)^*\Vol.
\]
where $L\in C^{\infty}(J^1Y)$. Then, 
\[
\ConstitutiveDensity=\brk{(\parJY_iL)\,\dxJY^i+((\parJY)_i^{\alpha}L)\,(\dxJY)_{\alpha}^i}\otimes_{J^1Y}(\pi^1)^*\Vol,
\]
It follows that 
\[
\ConstitutiveDensity_i=\parJY_iL\Textand\ConstitutiveDensity_i^{\alpha}=(\parJY)_i^{\alpha}L.
\]
} 

\section{Some Special Cases}\label{sec:Examples}
We present here a number of examples where the general settings assume particularly useful forms.
\subsection{Vector bundles\label{subsec:Vector-bundles}}
Consider the case where $(\bdl,\pi,\base)$ is a vector bundle. By
the natural isomorphism $T_{w}\vs\simeq\vs$ for any vector space
$\vs$, it follows that 
\begin{equation}
(VY)_{y}\simeq T_{y}(\bdl_{\pi(y)})\simeq\bdl_{\pi(y)}\fall y\in\bdl.
\end{equation}
It follows that for $p\in\base$,
\[
(\bk^*VY)_p \simeq (VY)_{\bk(p)} \simeq Y_{\pi(\bk(p))} = Y_p,
\]
i.e.,

\begin{equation}
\bk^*V\pi\simeq\bdl
\end{equation}
independently of $\bk$. As such, this theory may be referred to
as linear.


Next, since $(\bdl,\pi,\base)$ is a vector bundle, so is $(J^1Y,\pi^1,\base)$.
Thus, in analogy, for any section $\xi:\base\to J^1Y$, 
\begin{equation}
\xi^*VJ^1Y \simeq J^1\bdl.
\end{equation}

\subsection{Affine bundles}

Consider the case where $(\bdl,\pi,\base)$ is an affine bundle modeled
on the vector bundle $(E,\rho,\base)$ (see \cite[p. 48]{Saunders}).
It is implied that for every $y\in\bdl$, 
\begin{equation}
(VY)_{y}\simeq T_{y}(Y_{\pi(y)})\simeq E_{\pi(y)}.
\end{equation}
Hence, 
\begin{equation}
\bk^*VY\simeq E
\end{equation}
independently of $\bk$.
Likewise, by  Lemma \ref{lem:jetsPBs} , 
\begin{equation}
(j^1\bk)^*VJ^1Y \simeq J^1(\bk^*VY)\simeq J^1E,
\end{equation}
independently of $\bk$.

\subsection{The bundle of frames}

As mentioned in the introduction, the bundle of frames plays an important role in the geometric theory of continuous distribution of dislocations.
Consider the bundle of frames $(F\base,\pi,\base)$ whose fiber at
$p\in\base$ is $GL(\reals^{d},T_{p}\base$), the space of invertible
mappings $\reals^{d}\to T_{p}\base$ viewed as bases (frames) in $T_{p}\base$.
Since $GL(\reals^{d},T_{p}\base)$ is an open subset of the
vector space of linear mappings $\Hom(\reals^{d},T_{p}\base)$, $F\base$
is a sub-bundle of the vector bundle $\Hom(\base\times\reals^{d},T\base)$.
We may conclude that for any $y\in F\base$,
\[
(VY)_{y}\simeq T_{y}(GL(\reals^{d},T_{\pi(y)}\base))\simeq\Hom(\reals^{d},T_{\pi(y)}\base)
\]
and for any frame field $\bk:\base\to F\base$,
\begin{equation}
\bk^*VY\simeq\Hom(\base\times\reals^{d},T\base).
\end{equation}
Again, we may further conclude that
\begin{equation}
(j^1\bk)^*VJ^1Y\simeq J^1(\bk^*VY)\simeq J^1(\Hom(\reals^{d},T_{\pi(y)}\base)).
\end{equation}
%

\subsection{Form-conjugate forces and stresses}

Consider the case where $Y=\ext^{p}T^*\base$, where $1\leq p\leq d$. This is
a special case of \ref{subsec:Vector-bundles} above. It follows that
for any $\bk\in\Q$, $\bk^*VY\simeq\ext^{p}T^*\base$
and $(j^1\bk)^*VJ^1Y\simeq J^1(\ext^{p}T^*\base)$.
Differentiable variational stress densities fields along $j^1\bk$ are, therefore, sections of
$\Hom(J^1(\ext^{p}T^*\base),\ext^{d}T^*\base)$
and continuous traction stress density fields are sections of $\Hom(\ext^{p}T^*\base,\ext^{d-1}T^*\base)$.

A particularly simple case follows when the traction stress density field $\trst = P_{{\bk}}\circ \StressDensity$
is given in terms of a $(d-p-1)$-differential form $g$ by
\begin{equation}
\trst\circ v=g\wedge v
\end{equation}
for any $p$-differential form $v$ on $\base$. It follows that 
\[
\begin{split}
d(\trst\circ v) & =d(g\wedge v),\\
 & =dg\wedge v+(-1)^{d-p-1}g\wedge dv.
\end{split}
\]
Set,
\begin{equation}
\mathfrak{f}=dv,\qquad\mathfrak{J}=dg.\label{eq:Max1}
\end{equation}
It follows that, 
\begin{equation}
d\mathfrak{f}=0,\qquad d\mathfrak{J}=0.\label{eq:Max2}
\end{equation}
The representation of the force 
\[
\begin{split}\int_{\base}\StressDensity\circ j^1v & =\int_{\base}d((P_{\bk}\circ\StressDensity)\circ v)-\int_{\base}\diver_{\bk} \StressDensity\circ v,\\
 & =\int_{\base}d((P_{\bk}\circ \StressDensity)\circ v)+\int_{\base}\ForceDensity\circ v,
\end{split}
\]
may be written now as
\[
\int_{\base}\StressDensity\circ j^1v =
\int_{\base}(\mathfrak{J}\wedge v+(-1)^{d-p-1}g\wedge\mathfrak{f}+\ForceDensity).
\]
Assume that $d=4$, $p=1$ and consider the special case where $\ForceDensity=0$.
If we interpret $v$ as the $1$-form potential (vector potential)
of electrodynamics, the expression we obtained for the power, together
with equations (\ref{eq:Max1}, \ref{eq:Max2}) are the governing
equations of electrodynamics. Here, $\mathfrak{J}$ is interpreted
as the $4$-current density, $\mathfrak{f}$ is interpreted as the
Faraday $2$-form, and $g$ is interpreted as the Maxwell $2$-form.
As no metric structure is used, the equations are in the framework
of pre-metric electrodynamics (see \cite{Hehl-Itin06}).
In case the Lorentz metric is used, one can apply the Hodge star operator and write $ g=*\mathfrak f $ and obtain the usual Maxwell equations in vacuum, as for example in \cite[Chapter 3]{Gravitation}.

In the general case, one obtains a generalization of electrodynamics,
referred to as $p$-form electrodynamics, as in \cite{Henneaux1986,Henneaux1988}.
(For further details see \cite{Segev_e_d_2016}.)
A different theory in which form-conjugate forces appear is the theory
of dislocations. In the case of dislocations, however, the traction
stress does not have a simple form as assumed above. In fact, $d\mathfrak{f}=0$
implies that no dislocations are present (see \cite{EpsteinAndSegevDisloc2012,EpsSeg2015}). 

\subsection{Trivial fiber bundles}
Consider the case where the fiber bundle $(Y,\pi.\base)$ is trivial, that is 
\[
Y=\base\times\man, 
\]
 for some $m$-dimensional manifold $\man$. A section representing a configuration in this case is the graph of a mapping 
\(
\base\tto \man.
 \)
Thus, it is natural in the case of trivial bundles to identify the configuration with such a mapping
\[ 
\bk:\base\tto \man.
 \]
It follows that the configuration space $\Q$ is an open subset of the manifold of mappings $C^1(\base,\man)$.

Since $TY\simeq T\base\times TM$, $ VY=\ker T\pi $ may be identified naturally with $\base\times T\man$. It follows that a generalized velocity at $\bk$ may be represented by a mapping
\[  
v:\base\tto T\man,\quad \text{such that,}\quad  \tau_\man\circ v=\bk,
\]
or alternatively, with a section of the vector bundle
\[
(\bk^*T\man,\bk^*\tau_\man,\base). 
\]

The jet bundle $J^1Y$ can now be replaced by the jet space of mapping, $ J^1(\base,\man) $ (see \cite[Chapter 1]{michor80}). An element of $ J^1(\base,\man) $ is determined by the tangent mapping $ T_p\bk $ for some mapping $ \bk $; $ \strain $ may be identified with the space of continuous vector bundle morphisms $ T\base\to T\man $ over $ C^1 $-mappings $ \base\to\man $.

For further details regarding the forms assumed by continuous forces and stresses, see \cite{SegevRevMMAS2012}.

\subsection{Continuum mechanics on manifolds}

The previous example includes the case of continuum mechanics on manifolds. Here, the manifold $\man$ is interpreted as the physical space $S$. For continuum mechanics, it is customary to require that configurations satisfy the principle of material impenetrability which implies that $ \Q $ be the space of all embeddings $ \base\to S $. Indeed, the subset of embeddings is open in the manifold of mappings $ C^1(\base,S) $ (see \cite{Hirsch,michor80}).

In the classical case, both $ \base $ and $ S $ are 3-dimensional and $ S $ is a Euclidean space.

\subsection{Continua with microstructure}

Continuum mechanics of materials with microstructure falls under the case where $ Y $ has the structure of a Cartesian product $ \base\times\man $, however, the manifold $ M $ has the additional structure of a fiber bundle 
\[
(\man,\rho,S).
\]
The fiber over the point $ z\in S $ is interpreted as possible micro-states that a material point located at $ z $ may assume. Terms such as ``internal variables'', ``internal degrees of freedom'', ``order parameters'', etc. are also used for the micro-states.

A configuration of a body with microstructure is a mapping
\[ 
\bk :\base\tto\man 
\]
for which \[ 
\bk_0:=\rho\circ\bk:\base\to S
 \]
 is an embedding.
 
For bodies with microstructre, forces and stresses include components that are conjugate to the micro-velocities and their jets (see \cite{MicroStrucMMMAS94}).

A number of examples of continua with microstructure are discussed in \cite{mermin79,Capriz89}. These include mixtures, liquid crystals, Superfluid helium-3, Cosserat continua, etc.
\subsection{Group action}
We consider the case where for some Lie group $ G $ there is a smooth left action 
\[ 
\lambda : G\times Y\tto Y
\]
such that for all $ g\in G $,
\[  
\lambda_g:\{g\}\times Y \tto Y
\]
is a fiber bundle morphism over the identity of $ \base $. Thus, for $ p\in \base $, the image of $ \lambda_p:G\times \{p\}\to Y $ is contained in $  Y_p $.

The tangent mapping
\[  
T\lambda:TG\times TY\tto TY
\]
may be decomposed into partial tangents, and in particular
\[  
T_1\lambda:TG\times Y \tto TY
\] 
may be restricted to the Lie algebra, $ T_e G $, and we obtain
\[  
T_1\lambda_e:T_eG\times Y\tto TY.
\]
It is observed that the image of $ T_{1}\lambda_e $ is a subset of the vertical sub-bundle. It follows that for any $ \gamma\in T_eG $, $\bk\in\Q$, there is a generalized velocity
\[  
v_\gamma:\base\to VY,
\quad v_\gamma(p)=T_1\lambda_e(\gamma,\bk(p)).
\]

It would be natural, therefore, to define an \Emph{equilibrated force} as an element $ f\in T^*_{\bk}\Q $ such that 
\[  
f(v_\gamma)=0,\quad\text{for all}\quad\gamma\in T_eG
\]
(see also \cite{MicroStrucMMMAS94}).
Requiring that all forces on all subbodies be equilibrated restricts the class of stresses. For example, in the classical formulation of continuum mechanics in a Euclidean space, invariance under the Euclidean group implies that the components $ \mu_i $ of the variational stress vanish and the components $ \mu_i^\alpha $ satisfy the usual symmetry conditions for the first Piola-Kirchhoff stress.



\medskip
\Emph{Acknowledgement:} {RK was partially supported by the Israel Science Foundation (Grant No. 661/13), and by a grant from the Ministry of Science, Technology and Space, Israel and the Russian Foundation for Basic Research, the Russian Federation. RS was partially supported by the H. Greenhill Chair for Theoretical and Applied Mechanics and by the Pearlstone Center for Aeronautical Engineering Studies at Ben-Gurion University.}
%

\end{document}